\documentclass[10pt]{article}
\usepackage{amssymb}
 \usepackage{amsthm}

\usepackage{amsbsy}
\usepackage{amsmath,amssymb}
\usepackage{amsfonts}

\usepackage{lscape}
\usepackage{verbatim}
\usepackage{graphicx}
\usepackage{subfigure}
\usepackage{dsfont}
\usepackage{multirow}

\usepackage{color}

\usepackage{bm}
\usepackage{url}

\usepackage[authoryear]{natbib}

\def\frac#1#2{{\textstyle{#1\over#2}}}

\DeclareSymbolFont{AMSb}{U}{msb}{m}{n}
\DeclareMathSymbol{\Natural}{\mathbin}{AMSb}{"4E}
\DeclareMathSymbol{\Integer}{\mathbin}{AMSb}{"5A}
\DeclareMathSymbol{\Real}{\mathbin}{AMSb}{"52}
\DeclareMathSymbol{\Rational}{\mathbin}{AMSb}{"51}
\DeclareMathSymbol{\Imaginary}{\mathbin}{AMSb}{"49}
\DeclareMathSymbol{\Complex}{\mathbin}{AMSb}{"43} 
\def\bi{\begin{itemize}}
\def\ei{\end{itemize}}
\def\bd{\begin{description}}
\def\ed{\end{description}}
\def\ben{\begin{enumerate}}
\def\een{\end{enumerate}}

\def\hat#1{{\widehat{#1}}}

\def\2to{{\ {\buildrel 2\over \longrightarrow}\ }}

\def\I1ton{{$I_1,\ldots,I_n$}}
\def\X1ton{{$X_1,\ldots,X_n$}}
\def\Y1ton{{$Y_1,\ldots,Y_n$}}
\def\Z1ton{{$Z_1,\ldots,Z_n$}}
\def\R1ton{{$R_1,\ldots,R_n$}}
\def\e1ton{{$e_1,\ldots,e_n$}}
\def\t1ton{{$t_1,\ldots,t_n$}}
\def\x1ton{{$x_1,\ldots,x_n$}}
\def\y1ton{{$y_1,\ldots,y_n$}}
\def\z1ton{{$z_1,\ldots,z_n$}}
\def\r1ton{{$r_1,\ldots,r_n$}}

\newtheorem{definition}{Definition}

\newtheorem{proposition}{Proposition}

\usepackage{geometry}

\usepackage[finalnew]{trackchanges} 
\addeditor{min}
\addeditor{maj}

\begin{document}

\title{Modeling asymptotically independent spatial extremes based on Laplace
  random fields}

\author{Thomas Opitz, INRA, UR546 Biostatistics and Spatial Processes,
Thomas.Opitz@paca.inra.fr}

\maketitle

\begin{abstract}
We tackle  the modeling of threshold exceedances in
asymptotically independent stochastic processes by constructions based
on Laplace random fields. Defined as \add{mixtures of } Gaussian random fields
\change{scaled 
with a stochastic variable following an exponential distribution}
{with an exponential variable 
embedded for the variance}, these processes possess useful
asymptotic properties while remaining  statistically
convenient.  Univariate Laplace distribution tails are  part of the limiting generalized Pareto 
distributions for threshold exceedances. After 
normalizing marginal distributions in data, a standard Laplace 
field can be used to capture spatial dependence among extremes. 
Asymptotic properties of Laplace fields are explored and
compared to the classical framework of asymptotic dependence. 
Multivariate joint tail decay rates are
slower than for Gaussian fields with the same covariance structure; hence they provide more conservative
estimates of very extreme joint risks while maintaining asymptotic
independence. Statistical inference  is illustrated on
extreme wind gusts
in the Netherlands where a comparison to the Gaussian dependence model
shows a better goodness-of-fit\remove{ in terms of  Akaike's criterion}.  In this application we fit the
well-adapted 
Weibull distribution\add{, closely
related to the Laplace distribution,} as univariate tail model.\remove{, such that
the normalization of univariate tail distributions can be done
through a simple power transformation of data.}
\end{abstract}
{\bf Keywords:}
spatial extremes; threshold exceedances; asymptotic independence;  elliptical distribution; joint tail decay; wind speed.


\section{Introduction}



Extreme value analysis provides a toolbox for
modeling and estimating extreme events in univariate, multivariate,
spatial and spatiotemporal processes
\citep{Coles.2001,Beirlant.al.2004,Davison.Padoan.Ribatet.2012}. Principal objectives \add{of spatial extreme value analysis}
are the \add{spatial prediction of extremes and the } extrapolation of return levels and periods beyond the
historically observed range of data. A major distinction of dependence
types can be made between asymptotic dependence when 
$\lim_{u\uparrow 1}\mathrm{pr}(F_{X_2}(X_2)\geq u
\mid F_{X_1}(X_1) \geq u) >0$ for two random variables $X_1\sim
F_{X_1}$ and $X_2\sim F_{X_2}$  and asymptotic independence when the
limit is $0$, provided the limit exists. 
Asymptotic independence can arise in environmental and climatic
data for space
lags or  time lags when the most extreme events become
more and more isolated in time, space or space-time. For many
processes like  wind gust speed
or heavy rainfall such behavior seems
plausible owing to physical limits, and it is corroborated by empirical
findings \citep{Davison.al.2013,Thibaud.al.2013,Opitz.al.2015}.
 In this paper, our objective is  to construct asymptotically
independent spatial processes that are flexible and
tractable models with useful properties for modeling threshold
exceedances.

Models for asymptotically independent extremes must adequately capture
the joint tail decay rates in multivariate distributions.  
 A first in-depth analysis of joint tail decay was given by
\citet{Ledford.Tawn.1996,Ledford.Tawn.1997}. Closely related bivariate models
\citep{Ramos.Ledford.2009} provide flexibility in the joint tail, yet an
explicit definition of the probability density cannot be given when
only one component is extreme and the
generalization to higher dimensions suffers from the curse of
dimensionality. A more flexible characterization of
multivariate tail behavior was developed by \citet{Wadsworth.Tawn.2013}. 
Useful models pertaining to this framework are obtained by
inverting max-stable processes \citep{Wadsworth.Tawn.2012}, allowing for
composite likelihood inference. 
Another approach that spans both asymptotically dependent and
asymptotically independent data is presented by
\citet{Wadsworth.al.2014}, who model bivariate tails by assuming
independence among the radial and the angular variable in a
pseudo-polar representation. 

For lack of a unified theory of asymptotic independence, a variety of modeling
approaches have proven useful in practice. They usually suffer from at least one of the following
restrictions: joint tail decay rates are difficult to
characterize; standard inference methods like classical likelihood are not
available;  only bivariate models are tractable and useful; the generalization to the
infinite-dimensional, spatial setup is not possible; the 
univariate tail models prescribed by extreme value theory
do not directly appear  as marginal distributions in the model,
necessitating marginal pretransformations that are not natural in the
extreme value context. 

In the following, we present the novel Laplace model for multivariate and spatial
extremes. It provides a good compromise with respect to the aforementioned
potential shortcomings.  It is parametrized by a covariance function and closely
related to standard Gaussian processes \remove{through a random
  scaling} based on \add{embedding} an exponentially
distributed variable for the variance. The resulting univariate
distributions are of the Laplace type, and the univariate tails
correspond to valid generalized Pareto limits of threshold exceedances. Classical likelihood inference using
threshold exceedances is
straightforward. Joint tail decay rates and
conditional distributions can be characterized in various ways. We use the terminology of \emph{spatial
  processes} for simplicity's sake, but an
extension to the spatiotemporal context is possible through
spatiotemporal covariance functions. The notion of a multivariate
threshold exceedance is not uniquely
defined; here we will concentrate on four sensible choices for extreme value analysis: exceedances observed at
one fixed site $s$, exceedances of a linear combination of values at 
$D$ fixed sites,  or exceedances of either the maximum or minimum
value over $D$ fixed sites.  
 
Section 2 gives a short exposition of some aspects of classical extreme
value theory that are 
necessary to understand univariate tail models, their link to
standard Laplace marginal distributions  and the notion of
asymptotic independence. Section 3 treats  definition and
inference of  the asymptotically independent
Laplace model for threshold exceedances, whose tail behavior is
characterized and contrasted with the asymptotic dependence case. An
application to modeling of spatial wind gust extremes in the Netherlands in Section 4 is
followed by the concluding remarks in Section 5. 


\section{Extreme value theory}
Presentations going beyond this short account can be found in \citet{Resnick.1987,Beirlant.al.2004,deHaan.Ferreira.2006}.
\subsection{Univariate limit distributions}
\label{sec:unilim}

The fundamental limit relation of extreme value theory establishes a
three-parameter limit distribution for adequately rescaled maxima of
an independent and identically distributed sample $X_i$, $i=1,\ldots,n$:
if normalizing sequences $a_n>0$, $b_n$ exist such that 
\begin{equation}
\label{eq:doa}
\max_{i=1,\ldots,n} a_n^{-1}(X_i-b_n) \rightarrow Z, \quad n
\rightarrow \infty,
\end{equation}
with a nondegenerate random variable $Z$, then the
distribution of $Z$ is
necessarily of the generalized extreme value type with cdf
\begin{equation*}
F_Z(z) = \exp\left(-\left[1+\xi\frac{z-\mu}{\sigma}\right]_+^{-1/\xi}\right),
\end{equation*}
with parameters for position $\mu$, for scale $\sigma>0$ and for shape
$\xi$. When $\xi=0$, $F_Z(z)$ is defined as the corresponding limit
$\exp\left(-\exp[-z]\right)\, 1(z>0)$. 
The tail index $\xi$ is crucial to distinguish between Pareto-type
(polynomial) decay for $\xi>0$, exponential decay for $\xi=0$ to a
finite or infinite essential supremum and
inverse polynomial decay to a finite essential supremum for $\xi <0$.
By transforming $Z$ to
$Z^*=0.5\left(1+\xi[Z-\mu]/\sigma\right)^{1/\xi}$, we normalize to  a 
Fr\'echet distribution $F_{Z^*}(z)=\exp(-1/[2z])\, 1_{(0,\infty)}(z)$.

Equivalent formulations in terms of
threshold exceedances have proven useful  both from a
conceptual and practical perspective. 
By absorbing the values of $a_n$ and $b_n$ into the parameters
$\sigma$ and $\mu$, it is then useful to anchor practical modeling in the
tail approximation
\begin{equation}
\label{eq:univtailmodel}
\mathrm{pr}(X\geq x) \approx
\left(1+\xi\frac{x-\mu}{\sigma}\right)^{-1/\xi} = - \log F_{Z}(x), 
\end{equation}
valid for large $x$. When $\xi=0$, this tail probability
is $\mathrm{pr}(X\geq x) \approx \exp(-[x-\mu]/\sigma)$, corresponding to a
generalized exponential tail. By conditioning on an exceedance over a
high threshold $u$ with $u \geq \mu$ if $\xi\geq 0$ and $u < \mu$ otherwise, we get  
\begin{equation}
\label{eq:gpd}
\mathrm{pr}(X\leq x\mid X \geq u) \approx
1-\left(1+\xi\frac{x-u}{\sigma_u}\right)^{-1/\xi}, \quad
\sigma_u=\sigma+\xi(u-\mu), \quad x \geq u, 
\end{equation}
 where the right-hand side defines the cdf of the limiting \emph{generalized
Pareto distribution} for threshold exceedances $x-u$.

\subsection{Normalized marginal distributions}
When we treat the dependence structure in a multivariate setting, it
is convenient to normalize all univariate marginal distributions to
a  \add{parameter-free target
distribution that appears naturally in the extreme value context}. We will focus on standardizations that lead
to limit distributions in \eqref{eq:doa} of the Fr\'echet type
$F_{Z^*}(z)=\exp(-1/[2z])$ ($z>0$) with
tail index $1$ and scale $0.5$;  we call this standardization the
$*$-transformation in the following.
Here we  slightly modify the usual normalization found in the literature
(with scale $1$) in order to establish the link
to the Laplace distribution more easily in the following. \remove{Apart from
the modification of some constants without relevance to the asymptotic theory, this change
of scale leads to no differences in  the exposition.}
In the literature, we often find $X^*=1/[2(1-F_X(X))]$, leading to
 a Pareto target distribution with shape parameter $1$ and lower bound
 $0.5$ when $F_X$ is
continuous. 
Similarly, $X^*=-1/[2\log(F_X(X))]$ establishes the 
Fr\'echet target distribution $F_{X^*}(x)=F_{Z^*}(x)$ when $F_X$ is
continuous. 
It is possible to use more general \add{probability integral
  transforms} for $*$\remove{-transformations.   To preserve the
  asymptotic dependence structure,
 $*$ should fulfill the following three assumptions: (1) preservation of the ordering of observations: 
$X_1>X_2\Rightarrow X_1^*>X_2^*$; (2) nonnegativity: $\mathrm{pr}(X^*\geq 0) =1$;
(3) standard Pareto tail behavior: $2x\times \mathrm{pr}(X^*\geq x)
  \rightarrow 1$ for $x\rightarrow\infty$. } \add{, as long as
  they lead to nonnegativity and
  standard Pareto tail behavior, \emph{i.e.}, $\mathrm{pr}(X^*\geq 0) =1$ and $2x\times \mathrm{pr}(X^*\geq x)
  \rightarrow 1$ for $x\rightarrow\infty$.}

\remove{To preserve the ordering, it suffices to use probability integral
transforms.} Concerning  the behavior in the bulk of the
target distribution, there is no other constraint than
nonnegativity \remove{. When moving towards the limit,}\add{since} the normalization of
$X^*$ by 
$a_n=n$ drives all non-extreme values from the bulk to $0$
\add{in the limit.}
However, when modeling and estimating multivariate or spatial
extremal behavior from a finite sample of data, we often also need
to specify the bulk behavior in the target distribution, see
\citet{Coles.Tawn.1991,Wadsworth.Tawn.2012,deCarvalho.al.2014}, for
instance. 
For a  useful choice of  target distribution, we here 
focus on the following two properties:   (1) a decreasing density on $(0,\infty)$: $f_{X^*}(x_1) \geq
  f_{X^*}(x_2)$ for any $0< x_1 \leq x_2 < \infty$
(2) exact Pareto tail behavior: $\mathrm{pr}(X^*\geq x) = 1/(2x)$ for
  $x\geq u$ with  $u>0.5$ chosen as small as possible. \add{These conditions yield a $*$-transformation ensuring a distribution of  $X^*$  that is as close as possible to the asymptotic
    univariate limit since the convergence of maxima to
    $Z^*$ is equivalent to the convergence of the renormalized sample $\{ X^*_i/n,\, i=1,\ldots,n\}$
 to a Poisson limit process on $[0,\infty)$ with
    decreasing intensity
    $\mathrm{d}x /(2x^{2})$  associated to the intensity measure $\Lambda([x,\infty))=1/(2x)$. }
The only distribution to satisfy these
requirements is characterized by the density 
\begin{equation}
\label{eq:loglaplace}
f_{X^*}(x)=\begin{cases} 1/(2x^2) & \text{ if } x\geq 1, \\
0.5  & \text{ if }  0 \leq x < 1, \\
0  &\text{ otherwise. } 
\end{cases}
\end{equation}
\remove{\protect see \citet{Opitz.2013a} (in French) for a more profound discussion of this and other target distributions. }
This target distribution is a mixture with probability $0.5$ of a uniform distribution on
$(0,1)$ and a standard Pareto distribution on $[1,\infty)$. 
Its importance for the following developments comes from considering
the distribution of $\log X^*$, whose density defines the standard Laplace
distribution \citep{Kotz.al.2001}, 
\begin{equation*}
f_{\log X^*}(x) = 0.5\exp(-|x|). 
\end{equation*}
Therefore, the Laplace distribution can be
considered as a natural standard
marginal distribution for the modeling of multivariate extreme values.  

\subsection{Multivariate and spatial limits}
Univariate limits are readily generalized to the multivariate and
spatial domain. For a sequence of independent and identically distributed
copies $\{X_i(s)\}$ ($i=1,2,\ldots,n$) of a spatial process $\bm X=\{X(s)\}$, we need normalizing sequences $ a_n(s) >0$, $
b_n(s)$ such that the convergence of componentwise maxima holds in
terms of finite-dimensional distributions: 
\begin{equation}
\label{eq:mvdoa}
\max_{i=1,\ldots,n} a_n(s)^{-1}(X_i(s)-b_n(s)) \rightarrow Z(s), \quad n
\rightarrow \infty,
\end{equation}
with nondegenerate marginal distributions in the limiting max-stable
process $\bm Z=\{Z(s)\}$; we say that $\bm X$ is in the maximum domain of attraction of
$\bm Z$. The univariate convergence of marginal distributions
corresponds to the theory in Section \ref{sec:unilim}. 
We say that a vector $X(\bm s)=(X(s_1),\ldots,X(s_D))$ is
asymptotically independent if the components $Z(s_j)$, $j=1,\ldots,D$,
are mutually independent. 
The limiting dependence
structure of a multivariate vector can be characterized
for $*$-transformed margins. Convergence \eqref{eq:mvdoa} is equivalent to the convergence of all
univariate distributions in the sense of \eqref{eq:doa} together with
the following condition: a $(-1)$-homogeneous measure $\eta$ with
$\eta(t A)=t^{-1}\eta(A)$ exists such that,  for any set $A$ relatively compact in
$[\bm 0, \bm \infty]\setminus\{\bm 0\}\subset \overline{\mathbb{R}}^D$ with $\eta(\partial A)=0$, we have 
\begin{equation}
\label{eq:vague}
  t\,\mathrm{pr} ( X^*(\bm s)/t \in A) \rightarrow \eta(A), 
  \quad t\rightarrow\infty,
\end{equation}
for the marginally normalized vector $X^*(\bm s)$.
Loosely speaking, the vague convergence property \eqref{eq:vague}
means that  we can use the approximation $\mathrm{pr}( X^*(\bm s) \in A) \approx \eta(A)$ for
``extreme'' sets $A$ bounded far from the origin $\bm 0$. 
In practice, the homogeneity of $\eta$ provides a simple formula for the
extrapolation of extreme event probabilities to more extreme yet
hitherto unobserved events: 
$\mathrm{pr}( X^*(\bm s) \in
t A) \approx t^{-1}\mathrm{pr}( X^*(\bm s) \in A)$, for $t\geq 1$ and an
extreme event $A$. By transforming $X^*(\bm s)$ to $\log X^*(\bm s)$, this
relation can equivalently be written as $\mathrm{pr}(\log  X^*(\bm s) - t
\in A) \approx \exp(-t)\,\tilde{\eta}(A)$ with $\tilde{\eta} = \eta
\circ \exp$, $t\geq 0$ and $A \subset \mathbb{R}^D$ an extreme event. 
Finally, we point out  that $\eta((\bm 0, \bm \infty))=0$ when
$X(\bm s)$ is
asymptotically independent\add{, \emph{i.e.}, the limit measure $\eta$ has
a singular structure with all of its mass concentrated on the Euclidean
axes. }

\subsection{Tail dependence coefficients}
\label{sec:coeftail}
Bivariate tail dependence coefficients are useful summaries of 
tail dependence, and in a spatial context we can study their evolution
with respect to the spatial lag between two sites. The \emph{tail correlation coefficient} $\lambda$ of a bivariate random
vector $(X_1,X_2)$  is defined as 
\begin{align}
\lambda &= \lim_{u\uparrow 1} \mathrm{pr}(F_{X_2}(X_2)\geq u \mid F_{X_1}(X_1)\geq
u) \nonumber \\ &= \lim_{u\uparrow 1} \mathrm{pr}(F_{X_2}(X_2)\geq u, F_{X_1}(X_1)\geq
u)/(1-u) \in [0,1]. \label{eq:lambda}
\end{align}
It is well-defined and symmetric in $X_1$ and $X_2$ if $(X_1,X_2)$ is
in a bivariate maximum domain of attraction. Its value  is
$0$ if and only if $X_1$ and $X_2$ are asymptotically independent. 
If this is the case,  more information is provided by the  \emph{residual
  coefficient} $\rho\in [0,1]$ \citep{Ledford.Tawn.1996,Draisma.al.2004}, where 
\begin{equation}
\label{eq:rho}
\mathrm{pr}(X_1^*\geq
x, X_2^*\geq x) = \ell(x) x^{-1/\rho}
\end{equation}
for large $x$
with $\ell(x)>0$  a slowly varying function such that
$\ell(tx)/\ell(t)\rightarrow 1$ for $x>0$ and 
$t\rightarrow\infty$. 
When $\rho$ is positive, it corresponds to the tail index $\xi$ of
$\min(X_1^*,X_2^*)$. For independent variables $X_1$ and $X_2$, we
have 
$\rho=0.5$. We remark that $\rho$ is always equal to $1$ when $X_1$ and $X_2$ are
asymptotically dependent, such that it is of interest to study either $\lambda$
or $\rho$, depending on the presence of asymptotic dependence. 
For the bivariate Gaussian distribution, the residual coefficient is
$(1+\rho_{\mathrm{lin}})/2$, with $\rho_{\mathrm{lin}}$ the linear correlation coefficient, and covers the full range $[0,1]$ of theoretically
possible values when $\rho_{\mathrm{lin}}$ varies
from $-1$ to $1$.   In the following, we exclude the case $\rho=0$
corresponding to strong negative assocation of $X_1$ and $X_2$. 
The residual coefficient $\rho$ can also be defined more generally for
$D>2$ by 
\begin{equation}
\label{eq:rhomv}
\mathrm{pr}(X_1^*\geq
x, \ldots, X_D^*\geq x) = \ell(x) x^{-1/\rho}
\end{equation}
for large $x$ if an adequately defined slowly varying function $\ell$
exists. 
 
\subsection{From modeling to inference of extremal dependence}
\label{sec:modandinf}
Since there is no natural ordering of multivariate data, different
notions of threshold exceedances can be of interest with respect to
the application context but also  with respect to tractability of statistical
inference. 
We focus on four useful choices that are common in
the literature. Given a fixed threshold $u$, a vector $\bm x$ is a \emph{marginal 
  exceedance} in $x_1$ if $x_1\geq u$, it is a \emph{sum exceedance}
if $\sum_{j=1}^D x_j\geq u$, it is a \emph{max exceedance} if
$\max_{j=1,\ldots,D} x_j \geq u$, and it is a \emph{min exceedance} if
$\min_{j=1,\ldots,D} x_j \geq u$.   For the corresponding exceedance sets, we write
$A_{x_1}(u)=\{\bm x \mid x_1\geq u\}$,
$A_{\mathrm{sum}}(u)=\{\bm x \mid \sum_j x_j\geq u\}$, $A_{\mathrm{max}}(u)=\{\bm x \mid \max_j x_j\geq
u\}$ and $A_{\mathrm{min}}(u)=\{\bm x \mid \min_j x_j\geq
u\}$. The limit \eqref{eq:vague} holds for these sets when $u>0$.

Whereas models for asymptotically
dependent data have a strong theoretical foundation, there is no
natural model class for describing how asymptotically independent data
approach their limit. Using the limiting asymptotic independence would
usually be too simplistic (by imposing independence among
maxima of components, for instance), and it would considerably underestimate joint risks at the subasymptotic
level.\remove{\protect Moreover, the limiting Pareto processes for threshold
exceedances \citep{Ferreira.deHaan.2014,Dombry.Ribatet.2014,Thibaud.Opitz.2014} would  not be well
defined  when the domain is not finite.}  It is therefore sensible to focus on extremal dependence models that adequately
capture the rate of convergence towards asymptotic independence. 
A simple approach consists in combining an appropriate model for
univariate tails with a Gaussian dependence model. Known as \emph{Gaussian copula modeling} \citep{Davison.Padoan.Ribatet.2012}, this approach has been explored by \citet{Bortot.al.2000}
for oceanographic data 
and  by \citet{Renard.Long.2007}
for hydrological data, among others. Its extension to the spatial domain is
straightforward. 

For exploratory analysis, it is useful to investigate 
estimates of tail dependence coefficients based on the empirical
counterparts of \eqref{eq:lambda}  and \eqref{eq:rho}. In the spatial context, one
can study their variation with distance. For Gaussian copula modeling,
likelihood estimators $\hat{\rho}_{\mathrm{lin}}$ of the linear correlation $\rho_{\mathrm{lin}}$, based on
threshold exceedances, are readily defined and yield an estimate
$(1+\hat{\rho}_{\mathrm{lin}})/2$ of the residual coefficient. 
Likelihood estimation based on threshold exceedances is a standard
approach for both asymptotically dependent and asymptotically
independent data. After fixing an exceedance region $A$, we assume a
parametric distribution $F_A$ according to a parameter vector $\bm
\theta$ for exceedances $\bm X\mid \bm X \in A  \sim F_A$. 
The choice of threshold $u$ in $A$ is rarely
straightforward since it is usually a trade-off between bias and
variance. If a lower threshold decreases estimation variance by using
more information from data, one also has to take into account the rate of
convergence towards the asymptotic regimes used for modeling. 
The likelihood $L$ for jointly estimating
the exceedance probability $\mathrm{pr}(\bm X \in A)$ and the extremal behavior
is based on the censoring of data $\bm x$ not falling into $A$:
\begin{equation}
\label{eq:likeligeneral}
L(\bm \theta; \bm x) = 1_{A^C}(\bm x) \times(1- \mathrm{pr}(\bm X \in A) ) + 1_{A}(\bm
x) \mathrm{pr}(\bm X \in A)f_A(\bm x)
\end{equation}
where $f_A$ is the density of the exceedance distribution $F_A$.

\section{Properties and modeling of exceedances of Laplace random fields}
\subsection{Laplace random fields}
We denote $\bm W= \{W(s),s\in K\}$ a Gaussian random field defined on
a nonempty domain $K$. Multivariate distributions represent a special
case by setting $K=\{1,\ldots,D\}$ for $D\in \mathbb{N}$. 
We provide the constructive definition of Laplace random fields as a
\add{mixture of }
Gaussian field\change{ having exponential variance distribution}{s
  based on embedding an exponential variable with scale
$2$ for the variance of the Gaussian field}. This definition and basic properties\remove{, along with many other
characterizations of univariate and multivariate Laplace distributions,} can be
found in the monograph of \citet{Kotz.al.2001}. Multivariate Laplace
distributions are also discussed in \citet{Eltolft.al.2006}.
\begin{definition}[Laplace random field]
For $\bm W$ a centered Gaussian random field  and  a random variable
$Y \geq 0$ with $Y^2$ following an exponential distribution with scale
$2$ and independent of
$\bm W$, the random field 
\begin{equation}
\label{eq:laplace}
\bm X = \{ X(s), s \in K\} = \{Y W(s), s \in K \}
\end{equation}
is called \emph{Laplace random field} (or \emph{Laplace field} in
short).  We call $\bm X$ a \emph{standard Laplace field} if $\bm W$ is
a standard Gaussian field. 
If $\bm W$ has a covariance function or covariance matrix $\Sigma$, we write $\bm X
\sim \mathcal{L}(\Sigma)$. 
\end{definition}
\add{The distribution of $Y$ is known as Rayleigh distribution with parameter $1$. For standard Laplace vectors, we use the
  notation $ \Sigma^*$ to indicate that $\Sigma^*$ is a correlation
matrix. In the remainder of this paper, the default assumption is that
$\Sigma$ has full rank such that the inverse $\Sigma^{-1}$ is well
defined.}\remove{\protect We briefly recall properties of Laplace fields that are
straightforward from their univariate and  multivariate
characterizations \citep{Kotz.al.2001}.}
The univariate marginal distributions $F_{X(s)}$ of a standard Laplace
field $\bm X$ are symmetric with mean
$0$ and variance $2$,  and they have pdf
$f_{X(s)}(x) = 0.5\exp(-|x|)$
corresponding to the univariate Laplace distribution. 
The variable $Y$  has pdf  
$f_{Y}(y) = y \exp(-0.5y^2)$.
The covariance function of $\bm X$ is $2\Sigma$. 
If $\Sigma(\bm s)$ is the covariance matrix of $W(\bm s)
=(W(s_1),\ldots,W(s_D))$, the conditional distribution of $X(\bm
s)$ given $Y=y$ is Gaussian with covariance $y^2 \Sigma(\bm s)$. 
Using the notation $\nu = 0.5(2-D)$, the pdf of $X(\bm s)$ is 
\begin{equation}
\label{eq:mvlaplace}
f_{X(\bm s)}(\bm x) = \frac{2^{0.5}}{(2\pi)^{0.5D} |\Sigma(\bm s)|^{0.5}}
\left(0.25 \bm x'\Sigma(\bm s)^{-1} \bm x\right)^{0.5\nu}
K_{\nu}\left(\sqrt{\bm x' \Sigma(\bm s)^{-1} \bm x}\right).
\end{equation}
Here, $K_{\nu}$ is the modified Bessel
function of the third kind, see Equation (A.0.4) in
\citet{Kotz.al.2001}. 
Whereas marginal densities are continuous but not differentiable at
$0$, the joint density $f_{X(\bm s)}$ has a singularity at $\bm 0$ for
$D>1$ such that $f_{X(\bm s)}(\bm x)\rightarrow \infty$ when
$\|\bm x\|\rightarrow 0$. Since our interest will be directed towards
extreme events with $\|\bm x\|$ strictly separated from $0$, we can
neglect this 
particularity. 

Laplace vectors are part of the larger class of elliptically
contoured random vectors, which can be represented using a random
\emph{radial  variable} $R\geq 0$ and a $D\times D$ deterministic matrix $A$ that
operates on a random vector $\bm U$, independent of $R$,  with uniform distribution over the
Euclidean unit sphere in $\mathbb{R}^D$. Then 
$ R A \bm U$
defines an elliptic random vector with \emph{dispersion matrix}
$\Sigma=AA'$. 
For the multivariate Laplace distribution, we get \citep{Kotz.al.2001}
\begin{equation*}
f_R(r) = 2^{1-0.25D}\Gamma(0.5D)^{-1}r^{D-1}K_{0.5D-1}(r), \quad r>0.
\end{equation*}
\add{For Laplace vectors, we get $X(\bm s) = Y W(\bm s) = YR_{W(\bm s)} A \bm U$ with
  $R_{W(\bm s)}$ and $\Sigma(\bm s)=AA'$ associated to the elliptical
  Gaussian vector $W(\bm s)$. The dispersion matrices of the Gaussian
  vector and the Laplace vector are the same and the  random vectors
  differ only by their radial variable. }

\subsection{Asymptotic behavior of exceedances}
\label{sec:asymptotics}
The univariate standard Laplace tails\remove{(or \emph{standard half-exponential}
  tails)} with $\mathrm{pr}(X(s)\geq
x)=0.5\exp(-x)$ for $x>0$  correspond to the univariate tail model
\eqref{eq:univtailmodel} with $\xi=0$, $\sigma=1$ and $\mu=-\log 2$. For
exceedances above a fixed threshold $u>0$, we get the conditional
generalized Pareto distribution $F_{X(s)\mid X(s)\geq u}(x) = 1
- \exp(-(x-u))$ for $x\geq u$ according to \eqref{eq:gpd}. 

We now
characterize multivariate tail behavior of a Laplace vector denoted as $\bm
X=(X_1,\ldots,X_D) \sim \mathcal{L}(\Sigma)$ with 
$\Sigma=(\sigma_{j_1j_2})_{1\leq j_1,j_2\leq D}$. 
Due to \change{geometric sum-stability of elliptical
  distributions}{the elliptical structure of its
  distribution}, weighted sums of
the components of $\bm X$ are Laplace distributed: $\sum_{j=1}^D \omega_j X_j \sim
\mathcal{L}(\bm\omega'\Sigma\bm\omega)$,  and we have \remove{geometric} sum-stability
with respect to the generalized Pareto tail family with $\xi=0$, $\mu=-\log 2$ and
$\sigma_u=\sqrt{\bm \omega'\Sigma\bm\omega}$.
The mean of the components of $\bm X$ follows a $
\mathcal{L}(D^{-2}\sum_{1\leq j_1,j_2\leq D} \sigma_{j_1j_2})$
distribution.  When $\Sigma=\Sigma^*$, then  $\bm X^*=\exp(\bm X)$ has standard log-Laplace
margins \eqref{eq:loglaplace} whose shape parameter is $\xi=1$, and the
geometric mean of components has tail behavior leading to a Pareto 
exceedance distribution with shape parameter $\xi=D^{-2}\sum_{1\leq
  j_1,j_2\leq D} \sigma_{j_1j_2} \leq 1$ for values above
$u=1$. This is in contrast to the multivariate standard Pareto limits
for asymptotically dependent distributions, where sums of components are again
Pareto distributed with shape $\xi=1$, see
\citet{Falk.Guillou.2008}. 

A useful coefficient of tail dependence for Laplace vectors can be defined through the scale parameter $\sigma=(\sum_{1\leq
  j_1,j_2\leq D} \sigma^*_{j_1j_2})^{0.5}$  of the sum of components, where
$\sigma^*_{j_1j_2}$ are the entries of the correlation matrix
$\Sigma^*$ associated to $\Sigma$  in order to
remove the effect of the marginal scale. When $D=2$, this coefficient
is $\sqrt{2(1+\sigma^*_{12})}$ and varies between $0$
(for $\sigma^*_{12}=-1$) and $2$ (for $\sigma^*_{12}=1$). When $D>2$, it
is lower-bounded by $0$ and reaches its upper bound $D$ in the case of
complete dependence ($\sigma^*_{j_1j_2}=1$ for all $1\leq j_1,j_2 \leq
D)$. This coefficient allows us to see how pairwise coefficients
relate to higher-order coefficients, which is not straightforward for
the standard coefficients  $\lambda$ and $\rho$ of extreme value
theory defined in Section \ref{sec:coeftail}, see
\citet{Schlather.Tawn.2003} for inequalities related to $\lambda$ in a
multivariate context, and see \citet{Strokorb.2015} for 
spatial properties. 

\begin{proposition}[Joint tail decay rates]
\label{theor:mvtail}
The residual coefficient of $(X_1,X_2)\sim\mathcal{L}(\Sigma^*)$ with 
$\sigma^*_{12}=\rho_{\rm lin}$ is  $\rho=\sqrt{(1+\rho_{\rm lin})/2}$. 
The joint tail decay is  characterized by a
conditional limit for $x,y\in (0,1]$  when $u\rightarrow \infty$, 
\begin{equation}
\label{eq:bvsurv}
\mathrm{pr}(X_1> q_{1-x/u},X_2 > q_{1-y/u}\mid X_1> q_{1-1/u},X_2 >
q_{1-1/u}) \rightarrow (xy)^{1/(2\rho)}, 
\end{equation}
where $q_{p}$ is the quantile of the standard Laplace distribution
associated to the probability $p\in (0,1)$.
The multivariate residual coefficient defined in \eqref{eq:rhomv} is
$\rho=1/\sqrt{\bm e'(\Sigma^*)^{-1}\bm e}$ for a Laplace random vector $\bm
X\sim\mathcal{L}(\Sigma)$ with $\Sigma^*$ the correlation matrix
associated to $\Sigma$ and 
$\bm e = (1,\ldots,1)'$. 
\end{proposition}
\begin{proof}
For the bivariate results, we check the conditions stated
in Theorem 2.1 of \citet{Hashorva.2010}, which establishes joint tail
decay rates for a large class of bivariate elliptical random vectors. 
\add{\protect The theorem states the following: if a function
  $w(\cdot)$ and the so-called \emph{Weibull tail coefficient} $\theta
\geq 0$ exist such
that
\begin{equation*}
\lim_{r\rightarrow\infty}\frac{1-F_R(r+t/w(r))}{1-F_R(r)}=\exp(-t) \
  \forall t\in\mathbb{R}, \quad \lim_{t\rightarrow\infty} w(ct)/w(t)
  =c^{\theta-1} 
  \ \forall c>0,
\end{equation*}
then $\rho=([1+\rho_{\rm lin}]/2)^{\theta/2}$. 
In our case, }we exploit that the Bessel function of the third kind has asymptotic behavior $K_{\nu}(x) \sim \sqrt{\pi/(2x)}\exp(-x)$ when $x$ tends to infinity. 
Therefore, $f_R(r) \sim c_R \sqrt{r} \exp(-r)$ with a constant $c_R>0$.
Using l'H\^opital's rule, we obtain
\begin{equation*}
\lim_{r\rightarrow\infty}  \frac{1-F_R(r+t)}{1-F_R(r)} = \lim_{r\rightarrow \infty} \frac{\sqrt{r+t}\exp(-(r+t))}{\sqrt{r}\exp(-r)} = \exp(-t).
\end{equation*}
We have $w(t)= 1 =
t^{\theta-1}$ with $\theta=1$, yielding $\rho=\sqrt{(1+\rho_{\mathrm{lin}})/2}$.
The convergence \eqref{eq:bvsurv} is a special case of  Theorem 2.1
from \citet{Hashorva.2010}, where we exclude the cases $x > 1$ or $y>1$ to obtain our
slightly simpler 
formulation in terms of conditional probabilities. 
The general multivariate residual coefficient is obtained in Example 1
of \citet{Nolde.2014} with $\alpha=1$ for the Laplace distribution. 
\end{proof}
For given linear correlation coefficient $\rho_{\rm
  lin}$,  the bivariate residual coefficient of Laplace distributions is
the square root of the Gaussian equivalent. The slower joint tail
decay rate is due to the stochastic variance embedded in the Gaussian process. 
A bivariate Laplace vector with $\rho_{\rm lin}=0$ has  uncorrelated components
$X_1$ and $X_2$, but the
corresponding residual coefficient $1/\sqrt{2}$ is different from
$1/2$ corresponding to independent $X_1$ and $X_2$. 
The following proposition resumes the extrapolation of probabilities for exceedance
sets. 
\begin{proposition}[Extrapolation of  exceedance probabilities]
\label{cor:exc}
For $\bm X\sim \mathcal{L}(\Sigma)$ a $D$-dimensional Laplace vector, $u>0$ a threshold and $\bm{t}=(t,\ldots,t)>0$ defining a
translation $\bm X - \bm{t}$ of the vector, we observe:
\begin{align*}
\mathrm{pr}(\bm X - \bm{t} \in A_{\mathrm{sum}}(u)) &=
\exp\left(-\frac{D}{\sqrt{\bm e'
    \Sigma \bm e}} \times t\right) \mathrm{pr}(\bm X \in A_{\mathrm{sum}}(u)),
\\
\mathrm{pr}(\bm X - \bm{t} \in A_{x_1}(u)) &=
\exp\left(-t/\sqrt{\sigma_{11}}\right) \mathrm{pr}(\bm X \in
A_{x_1}(u)), \\
\mathrm{pr}(\bm X - \bm{t} \in A_{\mathrm{max}}(u)) & \sim
\exp\left(-t/\sigma\right) \mathrm{pr}(\bm X \in A_{\mathrm{max}}(u)),\quad u
\rightarrow \infty,
\end{align*}
where $\sigma=\sqrt{\max(\sigma_{11},\ldots,\sigma_{DD})}$ in the last
equation. 
When $\Sigma=\Sigma^*$ is a correlation matrix, we further have 
\begin{equation*}
\mathrm{pr}(\bm X - \bm{t} \in A_{\mathrm{min}}(u)) \sim
\exp\left(-\sqrt{\bm e'(\Sigma^*)^{-1}\bm e} \times t\right)
\mathrm{pr}(\bm X \in A_{\mathrm{min}}(u)), \quad u\rightarrow
\infty. 
\end{equation*}
\end{proposition}
\begin{proof}
\add{\protect 
The exact tail decay rate for sum exceedances follows
directly from the sum stability of the Laplace distribution, here with $\bm{\omega}=\bm{e}$:
\begin{align*}
\mathrm{pr}(\bm X - \bm t \in A_{\mathrm{sum}}(u)) &=
\mathrm{pr}\left(\sum_{j=1}^D X_j \geq Dt+u\right) \\ &= 0.5\exp\left(-(Dt+u)/\sqrt{\bm e'
  \Sigma \bm e}\right) \\ &= \exp\left(-Dt/\sqrt{\bm e'
  \Sigma \bm e}\right)\times \mathrm{pr}(\bm X \in A_{\mathrm{sum}}(u)).
\end{align*}
Similarly, the exact rate of marginal exceedances is obtained by
  setting $\bm \omega=(1,0,\ldots,0)$. 
To derive the formula for maxima, we first consider the case where 
$\sqrt{\sigma_{jj}}=\sigma>0$ is constant for $j=1,\ldots,D$. Then
$\bm X^*=\exp(\bm X/\sigma)$ is of standard log-Laplace type with
density \eqref{eq:loglaplace} and satisfies the convergence \eqref{eq:vague}.
By defining $t^*=\exp(t/\sigma)$ and $u^*=\exp(u/\sigma)$, we find that both of the expressions
$u^*\,\mathrm{pr}(\bm X^* \in A_{\max}(u^*))$ and $t^*u^*\, \mathrm{pr}(\bm X^*
\in A_{\max}(t^*u^*))$ tend to $\eta(A_{\max}(1))=D$ as $u^*\rightarrow
\infty$, where $\eta(A_{\max}(1))$ is  the
extremal coefficient whose value is $D$ under asymptotic
independence. Since $D>0$, we get $\mathrm{pr}(\bm X^*
\in A_{\max}(t^*u^*)) \sim (t^*)^{-1} \mathrm{pr}(\bm X^*
\in A_{\max}(u^*))$. In terms of $\bm X$, $t$ and $u$, we get the
stated convergence rate $\mathrm{pr}(\bm X - \bm{t} \in A_{\mathrm{max}}(u)) \sim
\exp\left(-t/\sigma\right) \mathrm{pr}(\bm X \in
A_{\mathrm{max}}(u))$. In the case where the diagonal values of $\Sigma$ are not
constant such that  $\sigma_{j_0j_0}=\max_{j=1,\ldots,D}
\sigma_{jj}$ for $j_0\in J$ with $1\leq |J| <D$, we start by observing
that 
\begin{equation}\label{eq:ineq} 
\mathrm{pr}\left(\max_{j_0\in J} X_{j_0} \geq u\right) \leq \mathrm{pr}\left(\bm X \in
A_{\max}(u)\right) \leq \mathrm{pr}\left(\max_{j_0\in J} X_{j_0} \geq u\right) +
\sum_{j\not\in J} \mathrm{pr}(X_j\geq u).
\end{equation}
Since $0\leq \mathrm{pr}(X_{j} \geq u) /\mathrm{pr}(\max_{j_0\in J}X_{j_0} \geq u)
\leq  \mathrm{pr}(X_{j} \geq u) /\mathrm{pr}(X_{j_1} \geq
u)\rightarrow 0$ as $u\rightarrow\infty$ for any $j_1\in J$ and
$j\not\in J$, dividing the inequality \eqref{eq:ineq} by
$\mathrm{pr}(\max_{j_0\in J} X_{j_0} \geq u)$ shows that  $\mathrm{pr}(\bm X \in
A_{\max}(u)) \sim \mathrm{pr}(\max_{j_0\in J} X_{j_0} \geq u)$ when
$u\rightarrow\infty$, and the convergence rate of maxima of $\bm X$ is the same
as the convergence rate of maxima of the subvector $\bm
X_J$. Applying the result for constant diagonal elements
$\sigma_{j_0j_0}$, $j_0\in J$ to $\bm X_J$  then proves the statement
for maxima. 
The asymptotic rate
for min exceedances is directly linked to  the definition of the multivariate residual
coefficient $\rho$ whose value is given in Proposition
\ref{theor:mvtail}.
}
\end{proof}
\change{\protect 
Since $\exp(A_{\rm sum}(u))$ and $\exp(A_{\rm min}(u))$ lie within
$(\bm 0, \bm \infty)$ where the standard scale limit measure $\eta$ in
\eqref{eq:vague} has no mass due
to the asymptotic independence of $\bm X$, the
decay rates are faster for these exceedances types than for $A_{x_1}$ and $A_{\rm max}$ where decay rates are determined by the
univariate marginal decay rate. 
}{\protect 
When $\Sigma=\Sigma^*$ such that margins of $\bm X$ are standard
  Laplace, then \eqref{eq:vague} holds for $\bm X^*=\exp(\bm X)$  and
  the limit measure $\eta$ is zero for the sets  $\exp(A_{\rm sum}(u))$
  and $\exp(A_{\rm min}(u))$ that lie within
$(\bm 0, \bm \infty)$. The
decay rates are faster for these exceedance types than for $A_{x_1}$ and $A_{\rm max}$ where decay rates are determined by the
 univariate marginal decay rate. Proposition \ref{cor:exc} shows that
$t\,\mathrm{pr}(\bm X^*/\tilde{t} \in A)$ tends to a finite positive limit when
choosing $\tilde{t}=t^{\sqrt{\bm e'\Sigma^* \bm e}/D}$ for
$A_{\mathrm{sum}}$ and $\tilde{t} =t^{1/\sqrt{\bm e'
    (\Sigma^*)^{-1}\bm e}}$ for $A_{\mathrm{min}}$.}

\add{\protect The sum stability of elliptical distributions further provides exact
formulas for related exceedance probabilities:  $\mathrm{pr}(\bm X \in
A_{\mathrm{sum}}(u)) = 0.5\exp(-u/\sqrt{\bm e'
  \Sigma \bm e})$ and $A_{x_1}(u)=0.5\exp(-u/\sqrt{\sigma_{11}})$. To 
avoid plain Monte-Carlo approaches for calculating probabilities
$\mathrm{pr}(\bm X \not\leq \bm x)=\mathrm{pr}(X_j>x_j\text{ for at
  least one } j)$ which can be too inaccurate or too
slow in certain cases, we can instead exploit tailor-made algorithms
for calculating multivariate normal probabilities
$\Phi_{\Sigma}(\bm x)$ \citep{Genz.Bretz.2009}. Univariate numerical
integration 
then yields more accurate values based on 
\begin{equation}\label{eq:excprob}
\mathrm{pr}(\bm X \not\leq \bm x)=1-
\int_{0}^\infty \Phi_{\Sigma}(\bm x/ y) f_Y(y) \, \mathrm{d}y.
\end{equation}
}
\add{\protect \subsection{Conditional distributions}}
\change{The following
proposition, which is the basis for conditional simulation of Laplace
vectors,  provides the conditional densities of $Y$ when conditioning on the value
of a subvector of the Laplace vector $\bm X$. }{
\protect The following
proposition provides a basis for conditional simulation of Laplace
random vectors conditional to the value
of a subvector. When $\bm X=(\bm X_1,\bm X_2)\sim \mathcal{L}(\Sigma)$ is a $D$-dimensional
Laplace vector with $\bm
X_1=(X_1,\ldots,X_d)$ for $d\in \{1,\ldots,D-1\}$, we write
$\Sigma_{ij}$, $1\leq i,j\leq 2$ for the corresponding blocks  of the
dispersion matrix $\Sigma$.
}
\begin{proposition}[Conditional distributions]
\label{theor:cond}
\leavevmode
\add{\protect 
\begin{enumerate}
\item  When $\bm X=Y\bm W\sim \mathcal{L}(\Sigma)$ with mixing
variable $Y\geq 0$, $Y^2\sim \mathrm{Exp}$ with scale $2$, conditioning on  $\bm
X=\bm x$ yields the conditional density of $Y$ given as
\begin{equation*}
f_{Y\mid \bm X=\bm x}(y ) = c(\bm x, \Sigma) \times  y^{-(D-1)}\exp\left(-0.5[y^2+\bm
  x'\Sigma^{-1}\bm x/y^2]\right), 
\end{equation*}
with the constant $c(\bm x, \Sigma)=
(2\pi)^{-0.5D}|\Sigma|^{-0.5}/f_{\bm X}(\bm x) > 0$ depending only on $\bm x$ and $\Sigma$.
For $D=1$ and $\Sigma=(1)$, we obtain the conditional density 
\begin{equation}\label{eq:ycondx}
f_{Y\mid X=x}(y ) = 2^{0.5}\pi^{-0.5}
\exp\left(-0.5[y^2+x^2/y^2]+x\right) 
\end{equation}
whose  mode is $\sqrt{x}$. 
\item  The conditional distribution of $\bm X_2 \mid \bm X_1=\bm x_1$ is
 elliptical with density
\begin{equation}
f_{\bm X_2\mid \bm X_1=\bm x_1}(\bm x) = c_0^{-1}|\tilde{\Sigma}|^{-0.5}g\left( (\bm x -
  \tilde{\bm \mu})' \tilde{\Sigma}^{-1}  (\bm x -
  \tilde{\bm \mu})+c_1 \right),
\end{equation}
where, using $\nu=0.5(2-D)$, 
\begin{align*}
g(r) &=\frac{2^{0.5}}{(2\pi)^{0.5D}}
(0.25\sqrt{r})^{0.5\nu}K_{\nu}(\sqrt{r}), \\
c_0 &=s_{D-d}\int_0^\infty g(r^2+c_1)r^{D-d-1}\, \mathrm{d}r, \quad
c_1=\bm x_2'\Sigma_{22}^{-1}\bm x_2, \\
\tilde{\mu} &= \Sigma_{21}\Sigma_{11}^{-1}\bm x_1, \quad \tilde{\Sigma}=\Sigma_{22}-\Sigma_{21}\Sigma_{11}^{-1}\Sigma_{12},
\end{align*}
and $s_n$ denotes the surface area of the unit ball in $\mathbb{R}^n$
with $s_1=1$ and $s_n=2\pi^{0.5n}/\Gamma(0.5n)$, $n > 1$. The radius
$\tilde{R}$ in the 
pseudo-polar representation $(\bm
X_2 \mid \bm X_1=\bm x_1) \overset{d}{=} \tilde{\mu}+\tilde{R}\tilde{A}\bm U$ with
$\tilde{\Sigma}=\tilde{A}\tilde{A}'$ has density
\begin{equation}
f_{\tilde{R}}(r) = s_{D-d}\,c_0^{-1} r^{D-d-1} g(r^2+c_1), \quad r > 0.
\end{equation}
\end{enumerate}
}
\end{proposition}
\begin{proof}
\add{1.)} The change of variables $(y,\bm w)\leadsto (y,\bm x)=(y,y\bm w)$ in
$f_{Y,\bm W}=f_Y\times f_{\bm W}$ yields 
\begin{equation*}
f_{Y,\bm X}(y,\bm
x)=y^{-(D-1)}\exp(-0.5y^2) \times (2\pi)^{-0.5D}|\Sigma|^{-0.5}\exp\left(-0.5[\bm
  x'\Sigma^{-1} \bm x]/y^2\right).
\end{equation*}
 Using the
formula $f_{Y\mid \bm X=\bm x}(y)=f_{Y,\bm X}(y,\bm x)/f_{\bm X}(\bm
x)$ leads to the given expressions.
\add{\protect \\ 2.) The formulas for conditional elliptical distributions follow from the related standard theory, see Corollary 5
and Section 4 of 
\citet{Cambanis.al.1981}, with $g(\cdot)$ given through the relation
$f_{\bm X}(\bm x)=|\Sigma|^{-0.5}g(\bm x'\Sigma^{-1}\bm x)$.
}
\end{proof}
\add{\protect One concludes from the density expression \eqref{eq:ycondx} in
  Proposition \ref{theor:cond} that the random variable $Y\mid X=x$ is
 concentrated around its mode $\sqrt{x}$. Indeed, the ratio $f_{Y\mid
  X=x}(x^{0.5+\varepsilon})/f_{Y\mid X=x}(x^{0.5}) =\exp(-0.5(x^{1+2\varepsilon}+x^{1-2\varepsilon})+x)$ tends to $0$ exponentially fast  for $\varepsilon>0$ when $x\rightarrow\infty$.
Therefore, the distribution of
the renormalized variable $Y/x\mid X=x$ converges to a point
mass in $0$ as $x$ tends to infinity, and for $X$ being large both
$Y$ and $W$ must be large simultaneously.} This is different from
asymptotically dependent normal scale mixtures $Y \bm  W$ characterized by
a mixture variable $Y$ with tail index $\xi>0$, where  the most extreme
events are due only to high values of $Y$, see \citet{Opitz.2013}. 

\change{\protect When we want to condition $\bm X$ on the values of one of its
subvectors, we can write 
\begin{equation*}
\bm X=(\bm X_1,\bm X_2)=Y(\bm W_1,\bm W_2) \sim
\mathcal{L}(\Sigma), \quad \Sigma=\begin{pmatrix} \Sigma_{11} &
  \Sigma_{12} \\ \Sigma_{21} & \Sigma_{22}\end{pmatrix},
\end{equation*}
and we use Proposition \ref{theor:cond} to propose an algorithm for the
conditional simulation of $\bm X_2\mid \bm X_1=\bm x_1$. }
{For simulating $\bm X_2\mid \bm X_1=\bm x_1$, we could either directly
  calculate and simulate the elements of the pseudo-polar representation
  $\tilde{\mu}+\tilde{R}\tilde{A}\bm U$, or we can exploit the latent
  Gaussian structure as follows.}
We first sample  the mixture variable $\tilde{y}$ according to the density $f_{ Y\mid
  \bm X_1=\bm x_1}$ with a  standard
approach like inverse transform sampling based on the
numerical integration of $f_{Y\mid \bm X_1=\bm x_1}$. Given $\tilde{y}$, we sample
$\tilde{\bm {w}}_2$ according to the 
conditional Gaussian distribution $\bm W_2 \mid \bm W_1 = \bm x_1/\tilde{y}$
with mean $\Sigma_{21}\Sigma_{11}^{-1}\bm x_1/\tilde{y}$ and
covariance matrix
$\tilde{\Sigma}$. 

\subsection{Likelihood inference}
We assume that the data process $\{X(s), s\in K\}$ has been observed at
$D$ sites $\bm s=(s_1,\ldots,s_D)$. For simplicity's sake, we further
assume that the marginal distributions of the observed random vector
$X(\bm s)$ are of standard
Laplace type. In practice, this requires marginal pretransformation
according to an adequate univariate model $F_{X(s)}$. We adopt the exceedance-based approach characterized through
the censored likelihood \eqref{eq:likeligeneral}.
We have to fix an exceedance set
$A\subset \mathbb{R}^D$ to characterize the region of extreme
events.
Since we are assuming extremal dependence according to
the Laplace random field model, we assume that the density $f_{ X(\bm s)}(\bm x)$ of $ X(\bm s)$ is equal to  the
multivariate Laplace density \eqref{eq:mvlaplace} for $\bm x \in
A$. Denote $p_A=\mathrm{pr}( X(\bm s) \in A)$ the exceedance probability for a
Laplace vector. 
Given a correlation model with (parametric)
correlation matrix $\Sigma^*(\bm s)$, the likelihood contribution of an observed
vector $ x(\bm s) = (x(s_1),\ldots,x(s_D))$ is 
\begin{equation}
\label{eq:lik}
L(\Sigma^*(\bm s); x(\bm s)) = 1(x(\bm s) \not\in A)\times
(1-p_A(\Sigma^*(\bm s))) + 1( x(\bm s)  \in A)\times
f_{X(\bm s)}(x(\bm s)). 
\end{equation}
When no tailor-made method exists for calculating the  exceedance
probability $p_A(\Sigma^*(\bm s))$, a simple
Monte-Carlo procedure can be applied by generating an independent and
identically distributed sample of $X(\bm s)$ and
by using the observed proportion of realizations  inside the exceedance
region $A$. We recommend a sample size of at least $10000$, leading to a
standard approximation error of around
$\sqrt{p_A}/\sqrt{10000}=\sqrt{p_A}/100$ for small $p_A$. 

\section{Application to wind gusts}




We illustrate modeling on daily maximum wind gust data from the Netherlands
collected from 14/11/1999 to 13/11/2008 for $30$ meteorological
stations, available for download from the
Royal Netherlands Meteorological Institute
(\url{www.knmi.nl}). Modeling extreme wind gusts is important for
applications like insurance risk
\citep{Brodin.Rootzen.2009,Mornet.al.2015}, forest damage
\citep{Pontailler.al.1997,Dhote.2005,Nagel.al.2006} or wind farming
\citep{Seguro.Lambert.2000,Steinkohl.al.2013}. 
\add{\protect Recent studies on similar data \citep{Engelke.al.2015,Einmahl.al.2015,Oesting.al.2015}
are based on max-stable models without challenging the assumption of
asymptotic dependence. 
For the present data, we will show that the asymptotically independent Laplace
model satisfactorily accomodates a situation that is inherent to
asymptotically independent data: tail correlation estimates
$\hat{\lambda}$ tend towards lower
values  when 
the tail fraction used for estimation is decreased. Moreover, when
using asymptotically dependent models, such data
behavior makes the choice of thresholds for inferential methods 
difficult.}
Other studies have shown that \change{asymptotically dependent models  are not
  appropriate}{asymptotically independent models 
are preferable} for wind speed data \citep{Ledford.Tawn.1996,Opitz.2013a}, hence  we
limit our analysis to asymptotic independence models by comparing models based on marginally transformed Laplace fields or
Gaussian fields. We
removed a small number of observation vectors with missing components
and retain observations for  $3241$ days. We  conducted a preliminary analysis that showed moderate to weak
day-to-day dependence of extreme wind gusts. Here
we focus only on spatial modeling of intra-day dependence and neglect
seasonal variations and 
clustering of extremes across several days.  

For estimating the parameters of the dependence structure with the
likelihood \eqref{eq:lik}, we use 
max exceedance sets $A_{\rm max}(u)$. This yields an
estimation strategy that is coherent for marginal modeling
(using observations above the marginal threshold $u$) and dependence
modeling (using observations within $A_{\rm max}(u)$ where at least
one marginal threshold is exceeded). 
Since the parameter vector comprising marginal parameters and dependence
parameters can be of dimension up to $7$ in our models, the joint estimation of marginal parameters and dependence
parameters through numerical maximization of the likelihood is prone
to numerical instabilities. We therefore estimate separately marginal
parameters (through the independence likelihood) and dependence
parameters. Since fitted univariate distributions often model
imperfectly the actual data distribution, which is not unusual for spatial data due
to their high dimensionality, we prefer to use the empirical
probability integral transform for transforming data to univariate
standard Laplace or standard Gaussian margins for estimating dependence
parameters. Computations have been carried out with the
\texttt{R} programming language \citep{Rcite}.

\subsection{Data exploration}
\label{sec:dataexpl}
Figure \ref{fig:stations} shows the spatial setup of
stations and information about sitewise quantiles for probabilities
$0.95$ and $0.99$. 
One can conjecture that wind gusts are slightly more extreme when sites are closer to
the coastline; this conclusion can be drawn for both of the 
investigated probability values \add{\protect and is illustrated by the color code
in Figure \ref{fig:stations} highlighting the $10$ highest and lowest
quantile values among the $30$ sites}.  Otherwise, quantiles seem stationary
in space.  \add{\protect Moreover, the ratio
$(q_{0.99}(s)/\overline{q}_{0.99})/(q_{0.95}(s)/\overline{q}_{0.95})$
with $\overline{q}_p$ the mean of the $p$-quantiles over all sites $s$
 is approximately constant, as can be seen from the
 information in Figure \ref{fig:stations} where symbols $+$ and
 $\times$ are of approximately the same size at each site. This
 justifies modeling the spatial nonstationarity of univariate tails
 through  a nonstationary scale parameter in the Weibull tail model 
 detailed in the following.}
Owing to the country's flat surface, it is reasonable to
assume that tail distributions vary only with respect to distance to
the sea.  Due to the rugged and irregular shape of the Netherlands' coastline, we use one
site's distance to the segments connecting the coastal sites as a proxy for its
distance to the sea. If a new site lies between this curve and
the actual coastline, we set its distance value to $0$. 

The left-hand plot in Figure \ref{fig:explor} shows empirical estimates of the tail correlation function
$\lambda(h)$ with respect to the distance between two sites, 
smoothed with local regression techniques (LOESS), for threshold
values $u$ given as $0.9$, $0.95$, $0.98$, $0.99$, $0.995$.  Estimates decrease
when $u$ increases, meaning that tail correlation weakens when we go
farther in the tail, which is  typical for asymptotically
independent data. Note that the
decision for or against asymptotic independence based on a finite data
sample can rarely be made with absolute certainty in practice. \add{\protect Although tail correlation estimates move closer to $0$
only very slowly when $u$ increases, goodness-of-fit checks  in Section \ref{sec:depmod}
 will confirm that the asymptotically independent Laplace model
can reproduce such behavior.} The assumption of
asymptotic independence is therefore plausible, and the right-hand
display of Figure \ref{fig:explor} shows pairwise estimates of the
residual coefficient $\rho$, calculated as
the Hill estimator \citep{Hill.1975,Draisma.al.2004} based on the largest $k=40$ order statistics of $\min(X^*(s),X^*(s+\Delta s))$ for all pairs of
sites. From the
pins in Figure \ref{fig:explor} indicating the orientation of the spatial lag of site pairs, it is difficult to
conclude on the presence of geometric anisotropy; we will therefore use model
selection tools to decide on this issue.
\begin{figure}
\centering
\includegraphics[width=10cm]{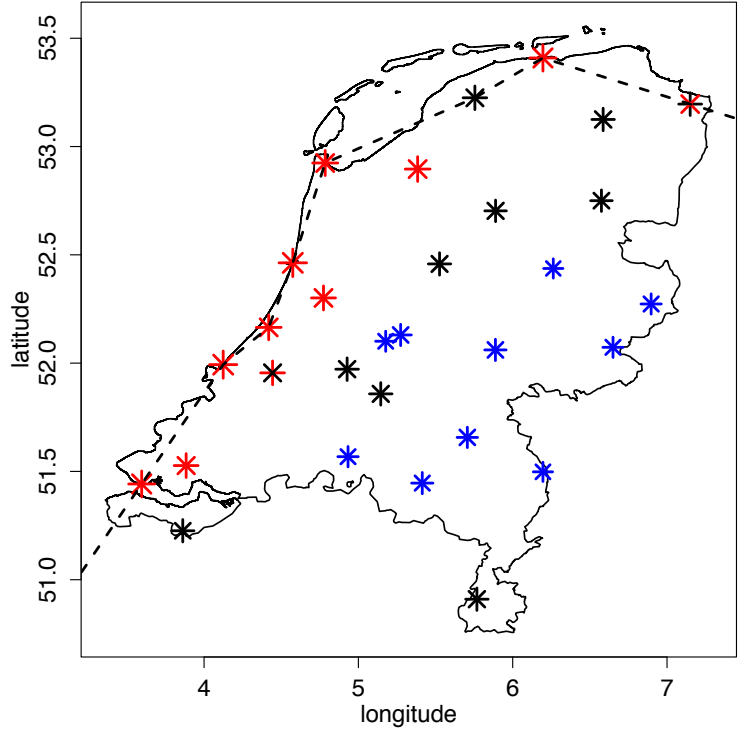}
\caption{Measurement sites for wind gusts in the Netherlands.  \emph{Coastal stations} are connected with a
  dashed and piecewise linear curve. Each
  site is marked with a $+$-symbol, whose size is proportional to the empirical
  $0.95$-quantile \add{divided by the mean of all
    such quantiles},
  and a $\times$-symbol, whose size is proportional \add{(with the same
  constant as for $0.95$-quantiles)} to the empirical
  $0.99$-quantiles\add{, divided by the mean of all
    such quantiles}. In each case, the sites associated to the $10$
  highest quantile values are marked in red, and the sites associated
  to the $10$ lowest quantile values are marked in blue.}
\label{fig:stations}
\end{figure}
\begin{figure}
\centering
\includegraphics[height=4cm]{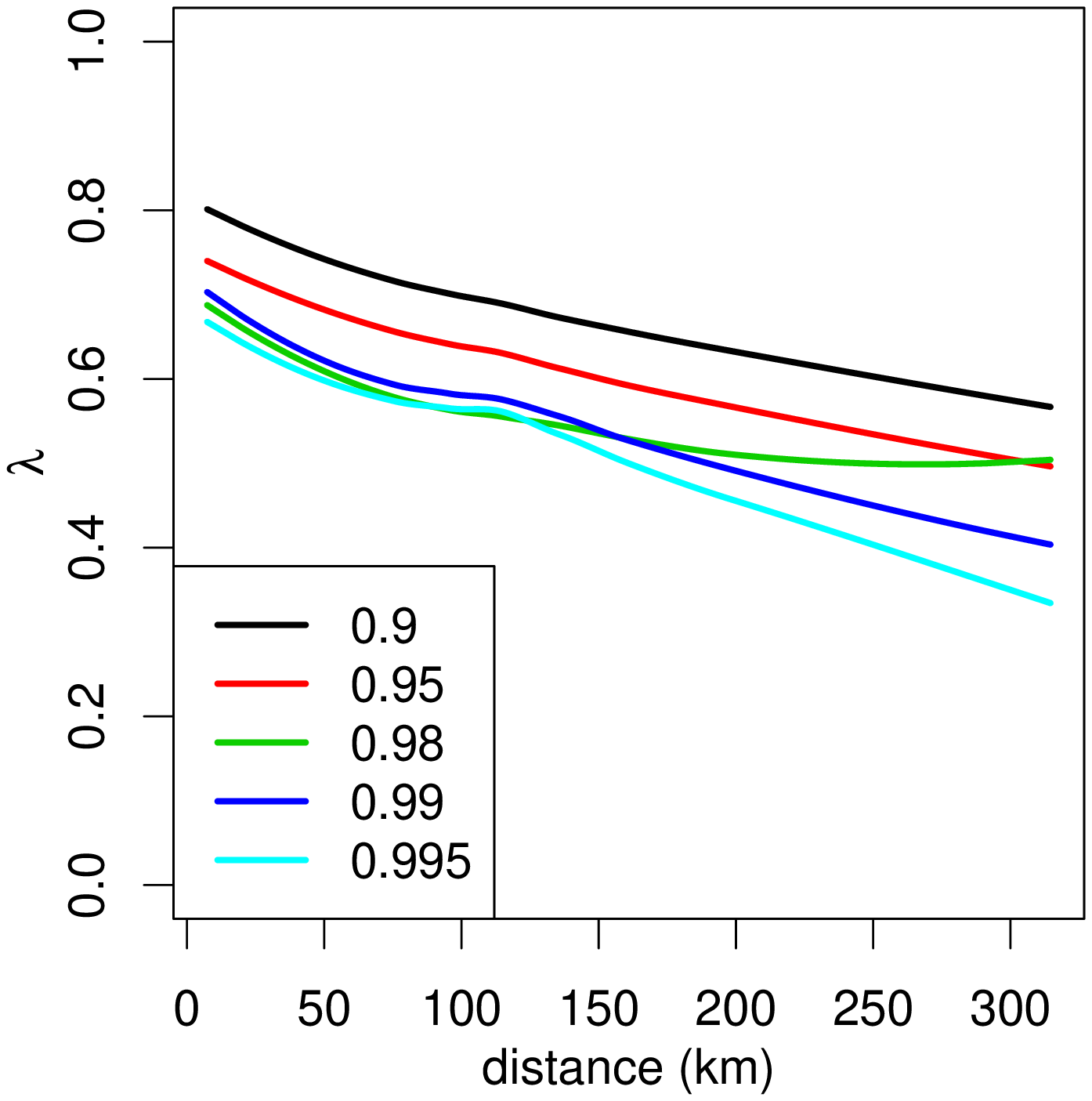}
\includegraphics[height=4cm]{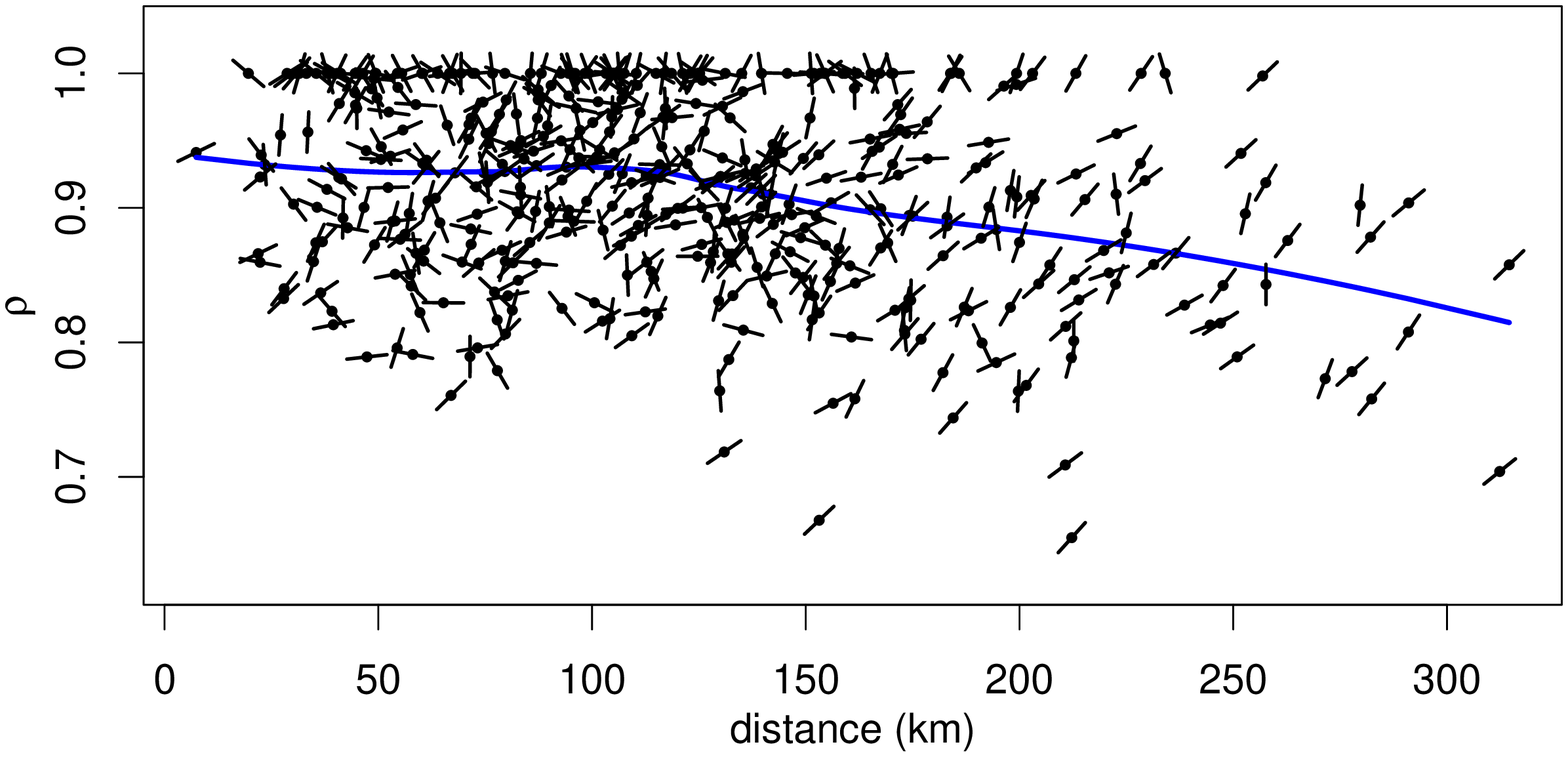}
\caption{Exploratory plots. Left: LOESS-smoothed empirical tail
  correlation functions $\lambda$ with respect to intersite
  distance, estimated for different thresholds $u$. Right: Empirical estimates of the residual
  coefficient $\rho$ with respect to intersite distance and
  local regression curve (LOESS); pins
  indicate the orientation of the spatial lag.}
\label{fig:explor}
\end{figure}
\subsection{Marginal modeling}
\label{sec:margmod}
The vast literature on wind
speed modeling has identified the Weibull distribution as the most
adequate candidate for univariate tails 
\citep{Stevens.Smulders.1979,Seguro.Lambert.2000,Akadaug.al.2009}. Therefore,
we will not fit the classical tail model \eqref{eq:univtailmodel} and
the related generalized Pareto distribution \eqref{eq:gpd} for exceedances to univariate tails,
but instead the Weibull distribution\add{ whose density is $x\mapsto 1_{[0,\infty)}(x)\gamma
\delta^{-\gamma} x^{\gamma-1}\exp[-\{x/\delta\}^\gamma]$ with
parameter constraints $\delta>0$ and $\gamma>0$}.  The Weibull distribution extends
the classical  tail model with shape
$\xi=0$ (limit of Gumbel type) and mean $\mu=0$ by introducing a supplementary shape
parameter $\gamma$;  
the classical model then applies to transformed data $x^{\gamma}$. 
We account for spatial variation by allowing the covariate ``distance to
the sea in \emph{km}'' to modify the
value of the Weibull scale parameter $\delta=\exp(\delta_0+\delta_1\times
\mathrm{covariate})$.\remove{\protect Based on our exploratory analysis in  Section
\ref{sec:dataexpl}, we choose $\mathrm{covariate}$ as the distance to
the sea in \emph{km}.} For a fixed threshold $u>0$, the contribution
of an observation $x$ to the independence likelihood of  univariate tail
parameters is
\begin{equation}
\label{eq:censweib}
1(x\leq u) \times \left(1-\exp\left[-\{x/\delta\}^\gamma\right]
\right)+1(x >  u) \times \gamma \delta^{-\gamma} x^{\gamma-1}\exp\left[-\{x/\delta\}^\gamma\right].
\end{equation}
\remove{where parameter constraints are $\delta>0$ and $\gamma>0$.} In
preliminary studies, we further
tried to include a location parameter in a way similar to the
classical tail
model \eqref{eq:univtailmodel}, \emph{i.e.} replacing $x$ by $x-\mu$
with $\mu < u$ in \eqref{eq:censweib}, but we could not detect any
substantial improvement in the goodness-of-fit and the numerical maximization
of the likelihood became unstable. 
After fixing $u$ to the empirical $0.975$-quantile calculated from all
observations, estimates and bootstrap-based confidence intervals of level $0.95$ are given as
follows: $\hat{\gamma}=1.72$($1.55,1.86$),
$\hat{\delta}_0=2.44$($2.29,2.52$) \emph{m/s} and
$\hat{\delta}_1=-0.0021$($-0.0026,-0.0019$)  \emph{m/s}, where
 confidence intervals have been calculated from a
block bootstrap sample (block
size $30$, sample size $100$). The effect of distance to the sea
$\delta_1$ is significantly different from $0$ \add{and
  indicates that the strength of wind gusts weakens when we move away
  from the coastline.}
We follow common practice in spatial extreme value modeling and use the
empirical distribution below the threshold by assuming dense
observations in the bulk region of the distribution
\citep{Coles.Tawn.1991,Wadsworth.Tawn.2012,Wadsworth.Tawn.2013}. 

\subsection{Dependence modeling}
\label{sec:depmod}
We consider Gaussian and Laplace dependence models with correlation
functions of exponential, stable or Mat\'ern type. For the Mat\'ern
model, we tried out a selection of regularity parameters $\nu\in\{0.1,
0.15, 0.2, 0.25, 0.3,
0.4, 1, 1.5, 2.5, 5\}$. 
We further allow geometric anisotropy to accommodate a potentially
different scale of dependence along one direction, for instance
stronger dependence orthogonal to the coastline for winds hitting land
from the sea.
By denoting $\theta\in[0,\pi)$ a rotation angle and $b\geq 1$ a stretching
along this direction, we replace the original bivariate column distance vector $\Delta
\bm s$ by the matrix product
\begin{equation*}
\begin{pmatrix} b & 0 \\ 0 & 1\end{pmatrix}\begin{pmatrix}\cos(\theta)
  & -\sin(\theta) \\ \sin(\theta) & \cos(\theta) \end{pmatrix} \Delta
\bm s = M\, \Delta \bm s.
\end{equation*}
Parameters $\theta$ and $b$ are estimated.  We use Akaike's
information criterion \textsc{AIC} to select the best model. To make
values of the censored likelihood \eqref{eq:lik} comparable for the
Gaussian and the Laplace model,  we first transform the original data $\bm x$  to
$\tilde{\bm x}$ with uniform margins  through the empirical transform  for each of
the $30$ sites. If $L(\bm \theta; \bm x)$ is the likelihood
\eqref{eq:lik} for either Laplace or Gaussian margins, we use instead
the likelihood $\tilde{L}$ in terms of $\tilde{\bm x}$, 
\begin{equation}
\label{eq:likdata}
\tilde{L}(\bm \theta;\tilde{\bm x}) = L(\bm \theta; F^{-1}(\tilde{\bm x}))\times
\prod_{i,j}(f(F^{-1}(\tilde{x}_{ij})))^{-1}, \quad j=1,\ldots,30, \quad i=1,\ldots,3241,
\end{equation}
where $F$ is either the standard Laplace or standard Gaussian cdf, and
$f$ is the corresponding density. 

Table \ref{tab:est} sums up the estimation results, where the highest \textsc{AIC}
values \add{\protect $2\log \tilde{L}(\hat{\bm \theta};\tilde{\bm x}) -
2\, \mathrm{dim}(\hat{\bm \theta})$} of both the Laplace and Gaussian Mat\'ern models were obtained
for the shape parameter $\nu=0.25$. As before, standard errors have
been calculated through the block bootstrap approach. When calculating \textsc{AIC}, we
consider Mat\'ern models with different values of $\nu$ as different
models, such that $\nu$ is not counted for the model dimension. 
The shape parameters retained for the stable and the
Mat\'ern class indicate nondifferentiable trajectories. Geometric anisotropy improves $\textsc{AIC}$, but the
improvement is relatively small compared to the differences between
the three classes of correlation functions. \add{Since $\hat{b}>1$, we
  find that spatial dependence is weakest along the
direction vector $(\cos(\hat{\theta}),\sin(\hat{\theta}))\approx
(0.58,0.81)$, here calculated for the Mat\'ern Laplace model, whereas
it is strongest when we move orthogonally to this direction. This direction approximately follows the
orientation of the coastline and we can conclude that winds hitting land
from the sea (or the other way round) are stronger dependent in space.}
Throughout, Laplace models outperform their Gaussian counterpart. Overall, we select the Laplace model
with the anisotropic Mat\'ern correlation which obtains the highest AIC value. 
The
Gaussian model with the same covariance family is second
best. \add{Since the isotropic Mat\'ern correlation function attains the 
 value $0.1$  approximately at distance 
$\sqrt{8\nu} \times \mathrm{scale}=\sqrt{2}\times\mathrm{scale}$, the estimated dependence is strong across the study region.}
For an
illustration of the goodness-of-fit, we look at the quantile-quantile
plots of the distribution of the sum of the observation vector above
the $0.95$-quantiles in
Figure \ref{fig:qq}. 
 The sum distribution is a good choice since its tail
decay behavior can be specified with an exact, nonasymptotic
expression for elliptical distributions, see the results of  Proposition
\ref{cor:exc} for the Laplace distribution. 
\add{\protect To facilitate the comparison of the displays for the Gaussian and the
Laplace model, theoretical quantiles of the  standard
Laplace distribution are used and the data are transformed accordingly. For instance,  the transformation of a data
vector $\tilde{\bm x}_i=\tilde{x}_{ij}$, $j=1,\ldots,30$ as defined
in \eqref{eq:lik} into its corresponding standard Laplace quantile $\tilde{q}_i$
is done as follows  for the Gaussian model:
\begin{equation*}
  \tilde{q}_i = F^{-1}\left(\Phi\left(\sum_{j=1}^{30}
    \Phi^{-1}(\tilde{x}_{ij})/\sqrt{\bm e' \hat{\Sigma}\bm e}\right)\right),
\end{equation*}
where $F$ is the standard Laplace cdf and $\Phi$
is the standard Gaussian cdf.}
In both cases, data points are aligned close to the
diagonal, with no strong differences between the Gaussian and the
Laplace model. Figure \ref{fit:bvresid} shows the difference of empirical and  theoretical 
residual coefficients for these two models, with empiricial
coefficients given as on the right-hand side of Figure
\ref{fig:explor}. The local mean of differences has slightly positive
values, which can at least partly be explained by second-order terms with respect to
the asymptotic decay rate.
\add{\protect Finally, Figure \ref{fit:tcf} contrasts
  the exact  
  conditional exceedance probabilities $\lambda_u(s_1,s_2) = \mathrm{pr}(F_{X(s_1)}(X(s_1)) \geq u \mid
F_{X(s_2)}(X(s_2)) \geq u)$ of the two models  with the corresponding
tail correlation estimates of data already presented on the left in Figure
\ref{fig:explor}. Whereas the Laplace model convincingly reproduces the observed
behavior in data, the Gaussian model is strongly biased towards too
low values. 
}



\begin{table}
\begin{center}
{\scriptsize
\begin{tabular}{llrrrrrrrrrr}
Covariance & Type & $\theta$ & $b$ & scale & shape & log-L & $\textsc{AIC}$ \\
\hline\hline
\multirow{ 4}{*}{exp} &G&$-$($-$)&$-$($-$)&369.0(22.44)&$-$($-$)&17420&34830\\

&L&$-$($-$)&$-$($-$)&376.6(18.43)&$-$($-$)&17620&35230 \\

&G&0.42(0.2)&1.14(0.06)&384.3(11.61)&$-$($-$)&17430&34850\\

&L&1.02(0.15)&1.09(0.05)&381.9(7.7)&$-$($-$)&17630&35250\\
\hline
\multirow{ 4}{*}{stable}
&G&$-$($-$)&$-$($-$)&3294(5.14)&0.53(0.01)&17740&35470\\

&L&$-$($-$)&$-$($-$)&3299(47.49)&0.50(0.01)&17890&35770\\

&G&0.89(0.11)&1.25(0.11)&3302(3.97)&0.54(0.02)&17750&35500\\

&L&1.14(0.1)&1.31(0.11)&3299(58.8)&0.52(0.01)&17930&35840\\
\hline
\multirow{ 4}{*}{Mat\'ern}
&G&$-$($-$)&$-$($-$)&3149(381.4)&0.25($-$)&17780&35560\\

&L&$-$($-$)&$-$($-$)&3087(297)&0.25($-$)&17900&35800\\

&G&0.87(0.12)&1.24(0.07)&3384(136.4)&0.25($-$)&18280&36550\\

&L&0.95(0.07)&1.41(0.09)&3698(121.7)&0.25($-$)&18440&36860\\

\end{tabular}
}
\end{center}
\caption{Estimated tail models for max exceedances of the Netherlands
  wind gust data. ``Type'' is either G for the Gaussian likelihood or L
  for the Laplace likelihood. Standard errors have been calculated
  through a block bootstrap approach. 
For the calculation of Akaike's
  criterion (\textsc{AIC}; here values are rounded to $4$ significant digits), the Mat\'ern shape parameter is not
  considered as an estimated parameter. }
\label{tab:est}
\end{table}

\begin{figure}
\centering
\includegraphics[width=5cm]{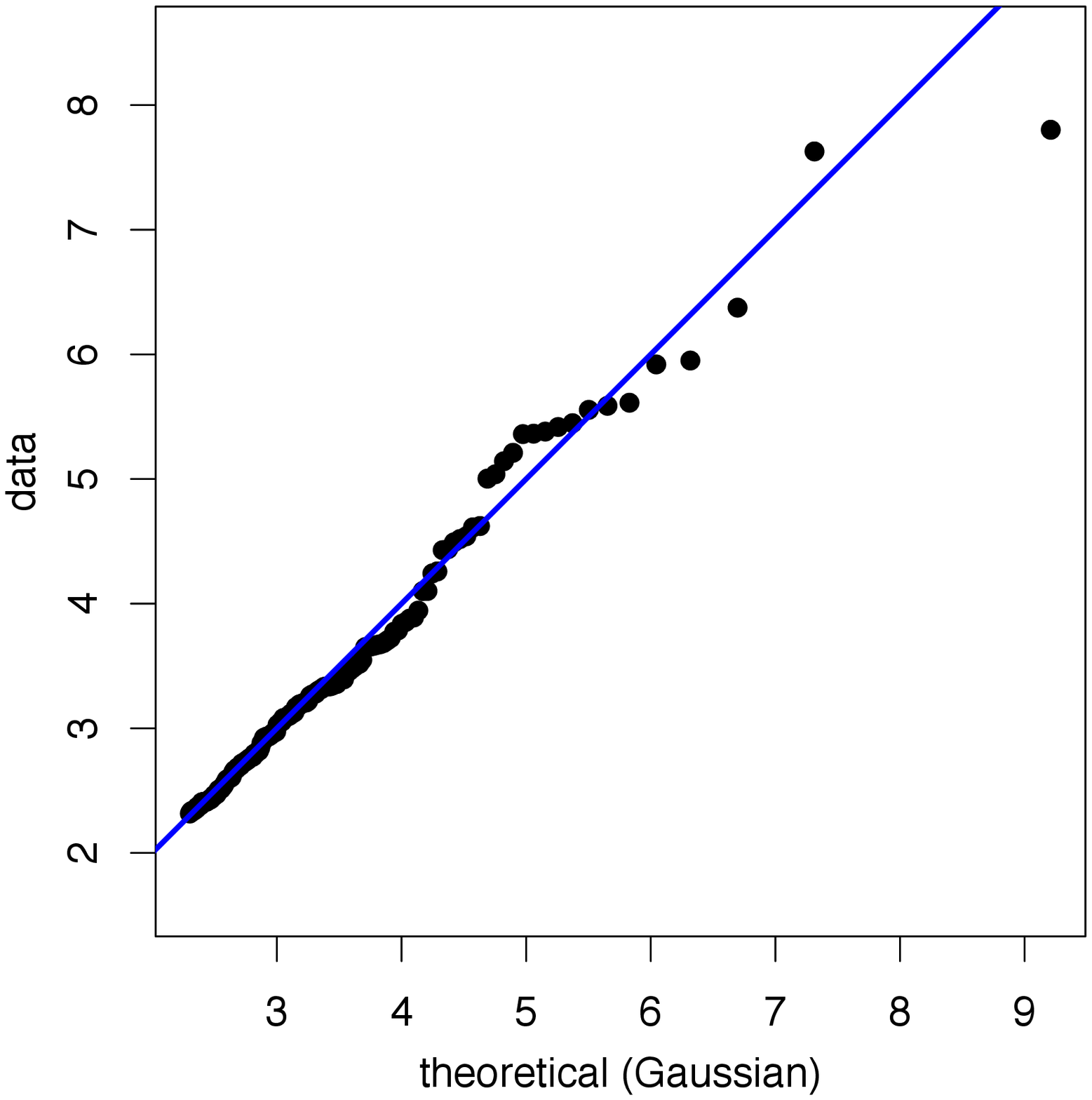}
\includegraphics[width=5cm]{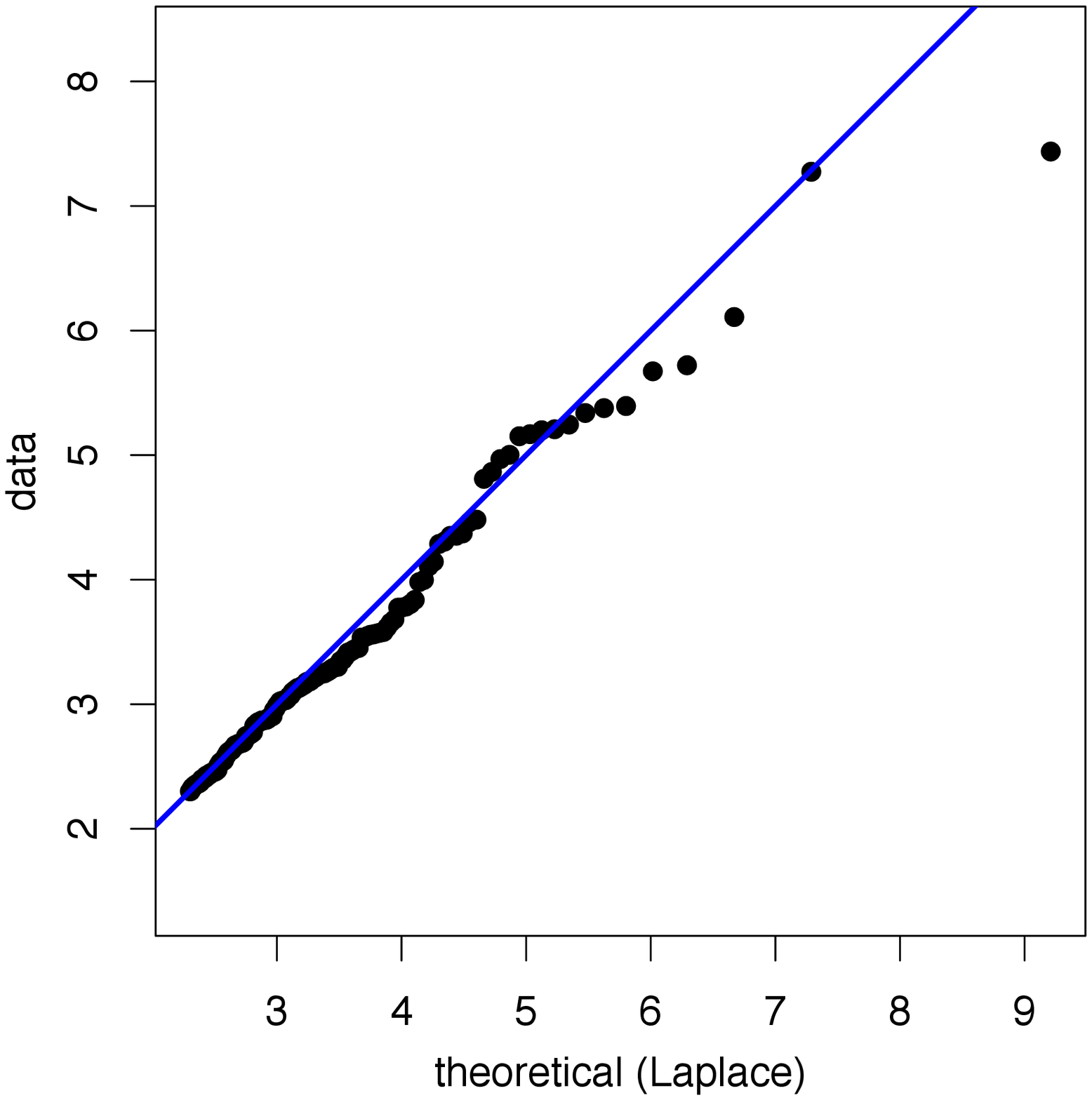}
\caption{Quantile-quantile plots for sum exceedances above the $0.95$-quantile in the Gaussian
  (left display) and the Laplace
  (right display) 
  tail dependence models with respect to the estimated Mat\'ern correlation models.  In both cases, the theoretical quantiles
  correspond to a standard Laplace tail and the  data quantiles are transformed
  accordingly.}
\label{fig:qq}
\end{figure}

\begin{figure}
\centering
\includegraphics[width=5cm]{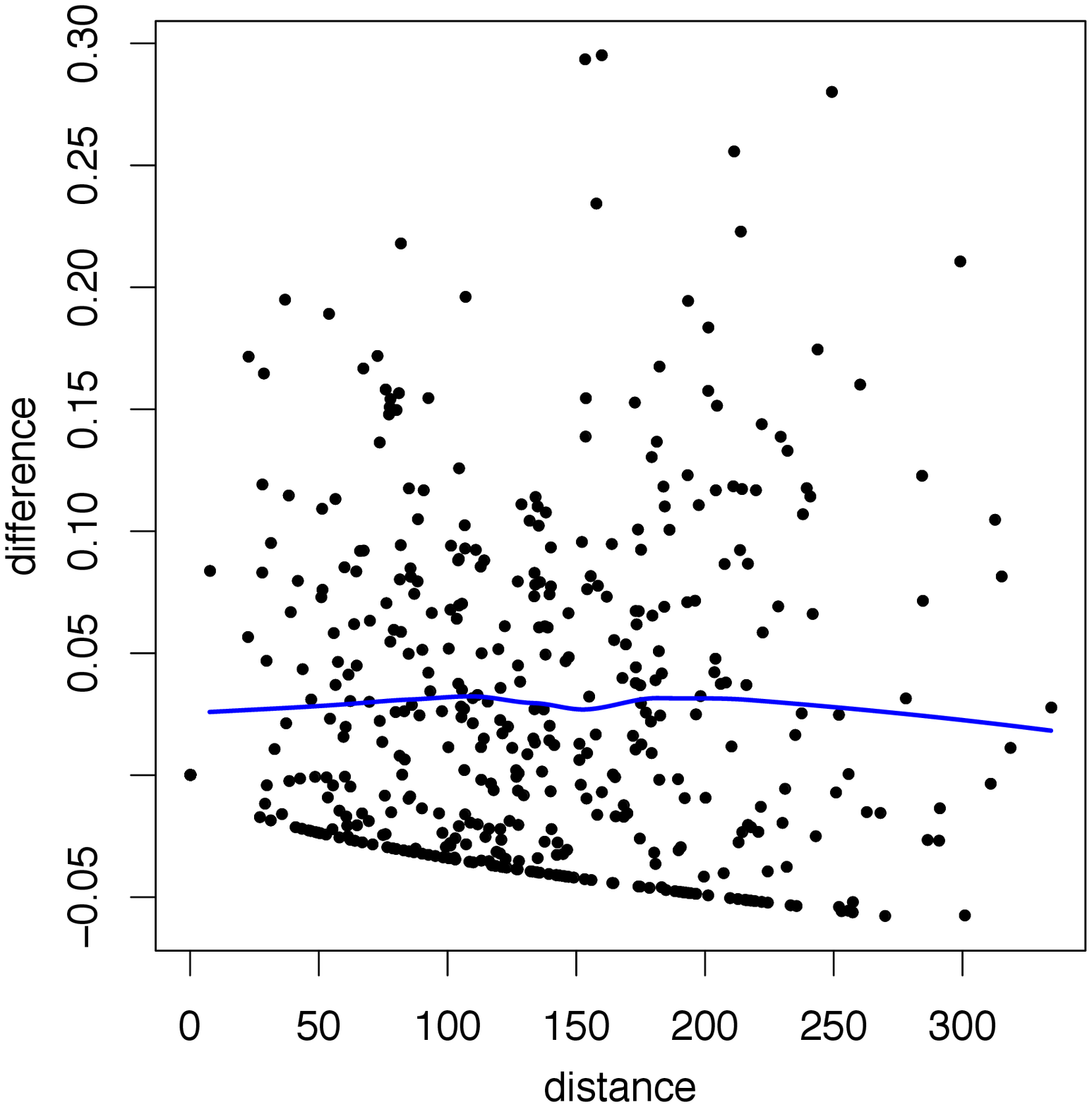}
\includegraphics[width=5cm]{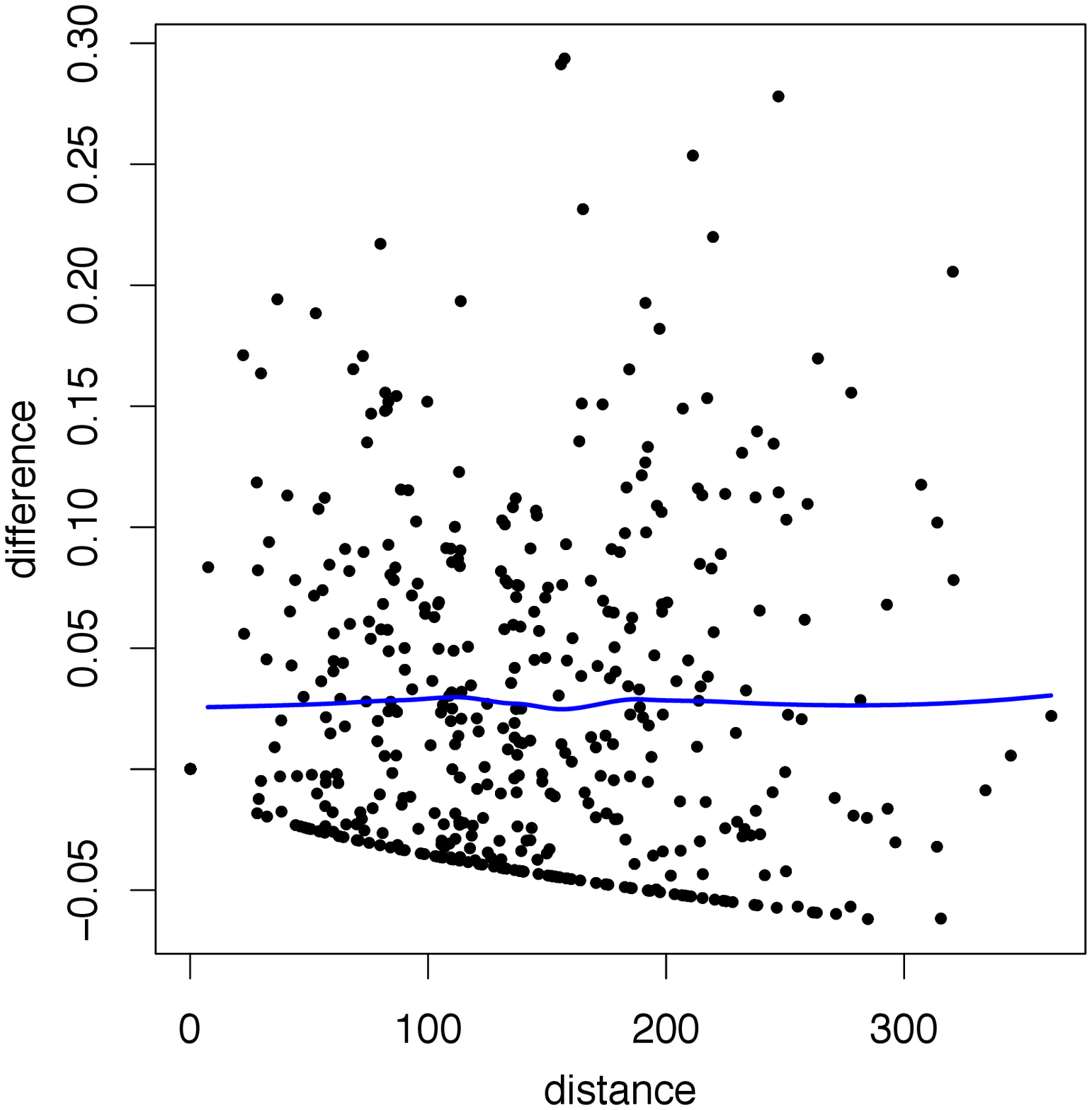}
\caption{Difference of pairwise empirical estimates and theoretical
  values  for the residual coefficient in the Gaussian
  (left) and the Laplace (right)
  tail dependence models with respect to the estimated Mat\'ern correlation models. Distances
  have been corrected for anisotropy according to the estimated
  anisotropy matrices such that $\mathrm{dist}=\|M\Delta \bm s\|_2$. }
\label{fit:bvresid}
\end{figure}

\add{\protect
\begin{figure}
\centering
\includegraphics[width=3.6cm]{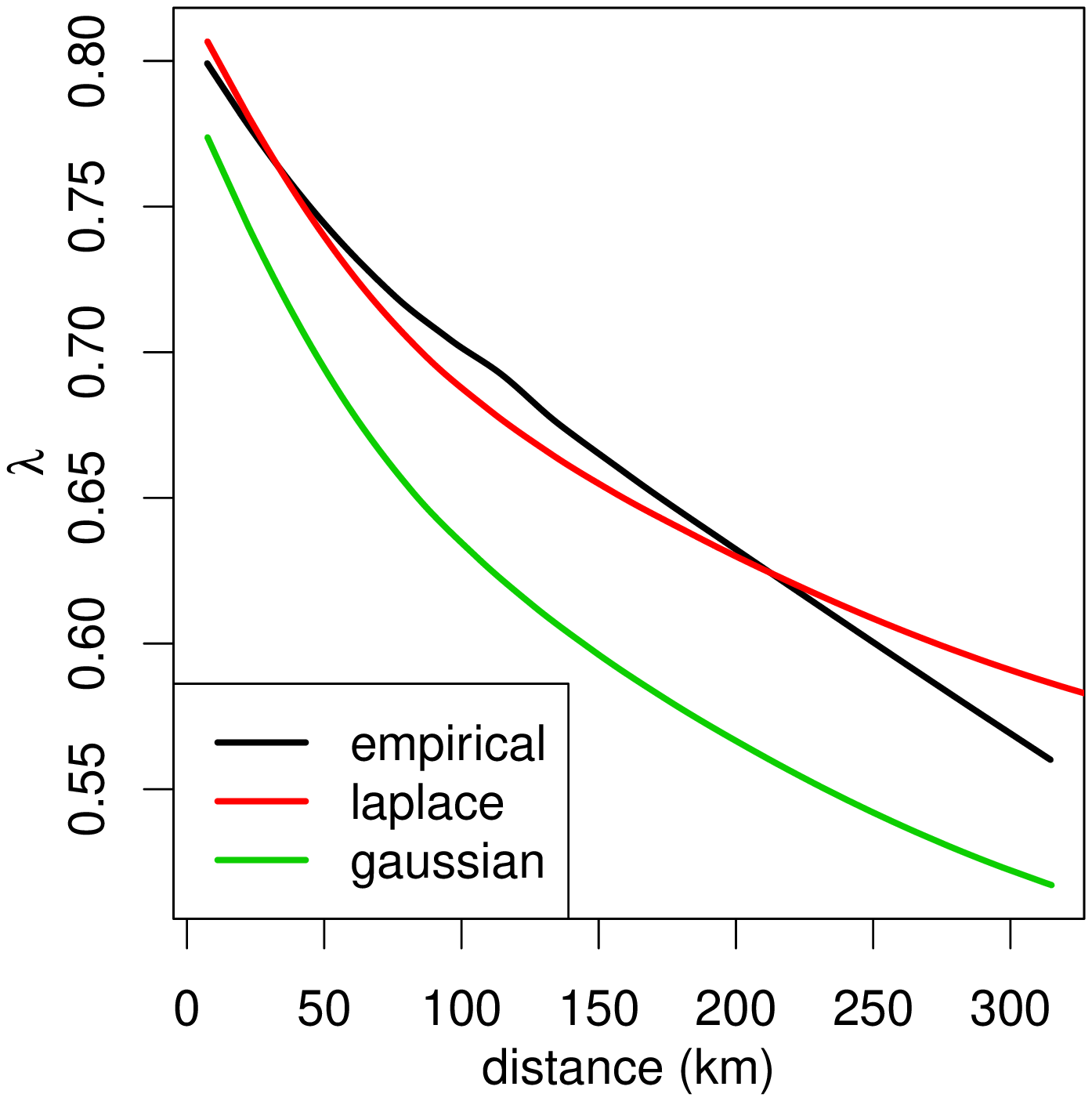}
\includegraphics[width=3.6cm]{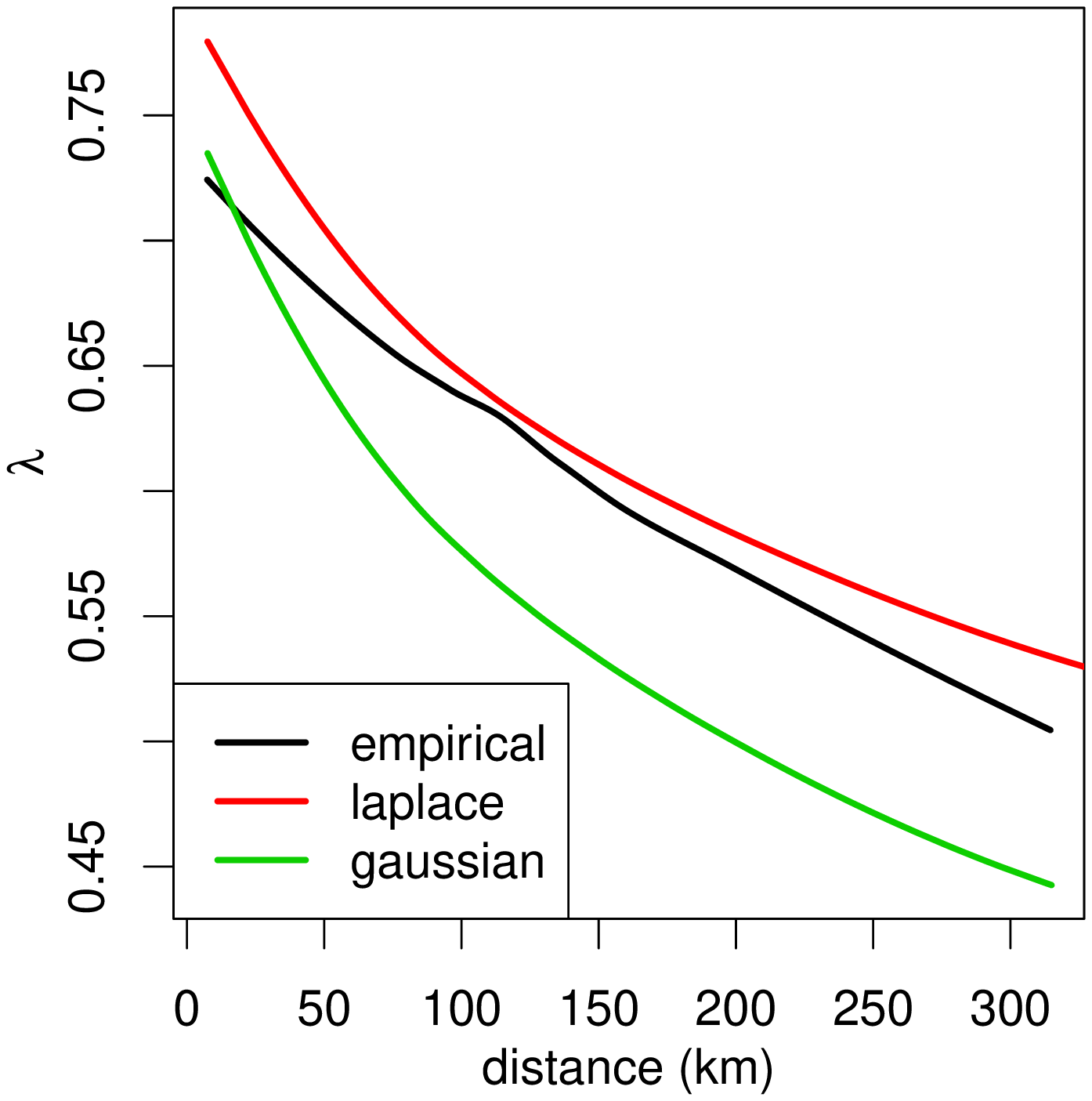}
\includegraphics[width=3.6cm]{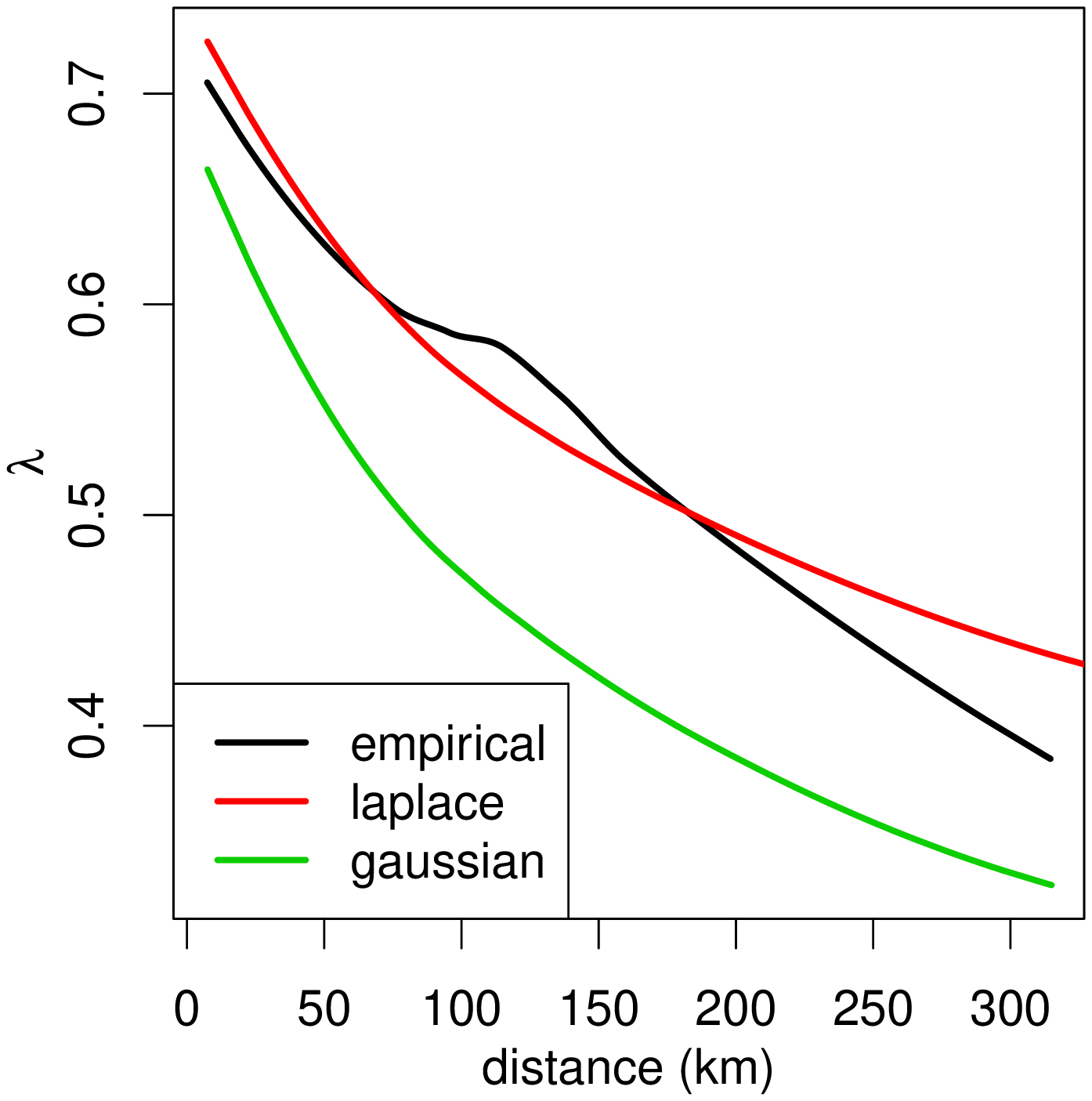}
\caption{Comparison of conditional exceedance probabilities $\lambda_u$ between data and fitted Laplace and
Gaussian models. From left to right for threshold values $u$ given by
$\{0.9,0.95,0.99\}$. The curves for data correspond to the
empirical smoothed tail correlation estimates presented on the left of
Figure \ref{fig:explor}.}
\label{fit:tcf}
\end{figure}
}
\subsection{Conditional simulations \add{ and return levels/periods}}
To illustrate conditional simulation for the  Laplace model according to the algorithm described at the end of Section \ref{sec:asymptotics}, Figure \ref{fig:condsim} shows
examples for two scenarios. In the first one, we condition on values
at the $30$ sites observed on 18 January 2007 during the Kyrill storm,
whose mean $32.4$ \emph{m/s} is the highest over the observation
period. 
In the second scenario, we condition on the estimated daily return level
$x_{\mathrm{cond}}=51.0$ \emph{m/s} of one of the coastal sites with respect to a
long return period of 10000 years, \emph{i.e.}, $1/(365.25\times 10^4) = \exp(-
(x_{\mathrm{cond}}/\hat{\delta})^{\hat{\gamma}})$. The gradient of
marginal scale with respect to distance to the sea is well
perceptible in both simulations.

\begin{figure}
\centering
\includegraphics[width=6cm]{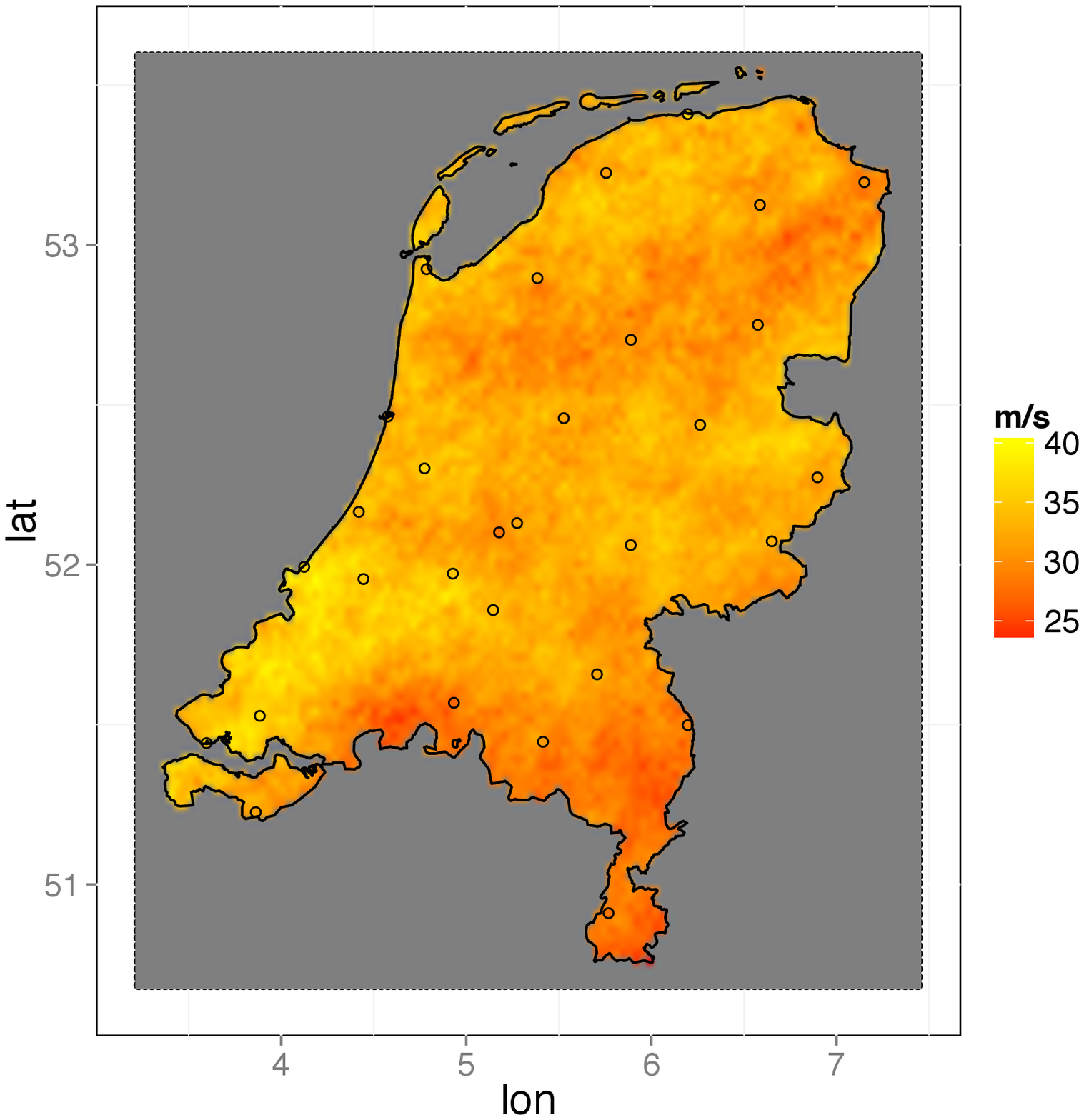}
\includegraphics[width=6cm]{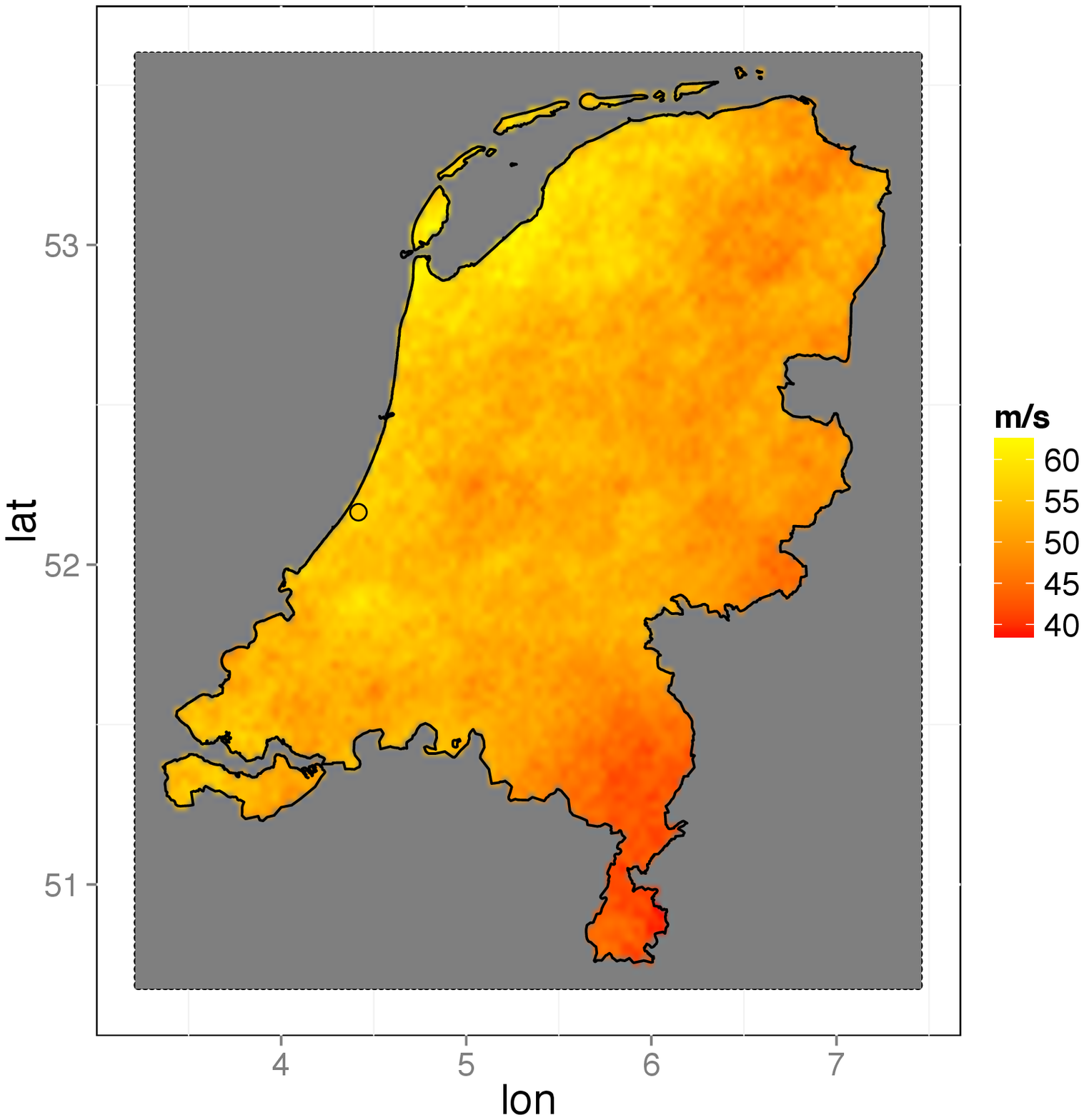}
\caption{Conditional simulations. Left: conditional to highest daily
  mean observation, observed on 18 January 2007. Right: conditional to 10000 year daily return
  period of one coastal site. Points indicate the conditioning sites
  and their values.}
\label{fig:condsim}
\end{figure}

\add{\protect Regarding the calculation of  return periods and levels  
  taking into account spatial dependence, we first
  consider the return period of an exceedance of the level $51$
  \emph{m/s} at at least one  point in the Netherlands territory.  Based on a fine spatial
  discretization $s_1,\ldots,s_{345}$ of the Netherlands and the
  corresponding correlation matrix $\hat{\Sigma}^*$ of the Laplace model, we calculate sitewise
  standard-scale values $\tilde{x}(s_j)=1-\exp(-x_0/\hat{\delta}(s_j))^{\hat{\gamma}}$,
  $j=1,\ldots, 345$ with $x_0=51$. Using formula \eqref{eq:excprob} for
  numerical integration, the return period in years is 
\begin{equation}\label{eq:RP}
\mathrm{RP}(x_0) = \frac{1}{ 365.25 \times \left( 1 - \int_0^\infty
\Phi_{\hat{\Sigma}^*}\left(\left( F^{-1}( \tilde{x}_1), \ldots,
    F^{-1}(\tilde{x}_{345})\right)/y\right) f_Y(y) \, \mathrm{d}y\right)}  = 376,
\end{equation}
where $F$ is the univariate standard Laplace distribution.
Therefore, considering the whole country
instead of a single site decreases the return period from $10000$ to
$376$ years. We can further determine the return level
associated to a return period of $10^c$ years ($c=0,1,2,3,4$) for the maximum observed over
the whole territory. We
numerically calculate the value $x_0^*$ for which the function
$x_0\mapsto |\mathrm{RP}(x_0)-10^c|$ attains its minimum $0$,
yielding values (in \emph{m/s})  $37.3$, $44.0$, $49.9$, $55.1$, $ 60.0$ for 
 $c=0,1,2,3,4$ respectively. 
}

\section{Discussion and perspectives}

Due to its generalized Pareto tails and the \remove{geometric} sum stability of
its multivariate elliptical distributions, the Laplace tail model has advantageous properties for
extreme value analysis while remaining close to the Gaussian processes
used in classical
geostatistics. For conditional and unconditional simulation in high
dimensions, the well-studied Gaussian techniques are
the main ingredient. Standard maximum likelihood inference
is possible.
In contrast to the Gaussian copula
model, it is easier to
interpret in the extreme value context \add{since the elliptical structure
  arises for Laplace marginal distributions which, unlike the
  univariate Gaussian, are more natural for
  extremes.}
Moreover,  the
standardization to marginal \change{half
exponential}{Laplace} distributions yields an
exponential tail decay rate that is relatively close to the actually
observed tail decay rates in environmental processes like wind speeds
or precipitations.  

The goodness-of-fit checks for wind gust data
revealed \change{no striking difference
between Gaussian and Laplace dependence structures, yet we can point
out a considerably better fit of the Laplace model  in
terms of Akaike's criterion.}{ a better fit of the Laplace model   in
terms of Akaike's criterion and bivariate joint tail decay behavior.}
We did not take into account temporal dependence among extremes. A simple
possibility for capturing clustering behavior in time would be to
model the time series of a summary statistics like the sum of the
marginally normalized
observations at the observed time steps with models akin to the exponential ARMA
processes of \citet{Lawrance.Lewis.1980}. 

The Gaussian scale mixture
construction of Laplace models
opens perspectives to more complex models. The latent exponential
variance variables
could be dependent  over time or vary over space, which would  then necessitate
a Bayesian framework to handle latent variables. 
\add{Finally, other than exponential variables $Y$ could be
  substituted for the variance of the Gaussian field. Such modeling
  would lead to 
  higher flexibility in the extremal dependence structure with respect to tail decay rates, yet densities
  and cdfs are usually not available in closed form, and typically the resulting marginal
  distributions  are not natural in the context of extremes.}

\section{Acknowledgements}
The author is grateful to an editor and to two referees for numerous
comments that helped improving the manuscript. 

 \bibliographystyle{elsarticle-harv} 
 \bibliography{ref}


\end{document}